\documentclass{article}

\usepackage{microtype}
\usepackage{graphicx}
\usepackage{caption}
\usepackage{subcaption}
\usepackage{booktabs} 
\usepackage[normalem]{ulem}

\usepackage[accepted]{icml2021custom}

\usepackage{enumerate}
\usepackage{amssymb}
\usepackage{amsmath}
\usepackage{nicefrac}
\usepackage{graphicx}
\usepackage{slashbox}
\usepackage{tikz,pgfplots}
\usepackage[breaklinks,bookmarks=false]{hyperref}

\usepackage{algorithm}

\usepackage{bbm}
\usepackage[amsmath,amsthm,hyperref,thmmarks]{ntheorem}
\usepackage[capitalise,noabbrev,nameinlink]{cleveref}
\usepackage{adjustbox}
\usepackage{multirow}
\usepackage{xspace}

\usepackage{thm-restate}

\newtheorem{theorem}{Theorem}[section]

\newtheorem{lemma}[theorem]{Lemma}

\newtheorem{definition}[theorem]{Definition}
\newtheorem{fact}[theorem]{Fact}

\newtheorem{remark}[theorem]{Remark}

\qedsymbol{\openbox}

\newcommand{\alglinelabel}{%
  \addtocounter{ALC@line}{-1}
  \refstepcounter{ALC@line}
  \label
}

\DeclareMathOperator{\dist}{dist}

\DeclareMathOperator{\cost}{cost}
\DeclareMathOperator{\poly}{poly}
\DeclareMathOperator{\OPT}{OPT}
\DeclareMathOperator{\OUR}{OUR}
\newcommand{\cO}{\mathcal{O}}
\newcommand{\tO}{\widetilde{\mathcal{O}}}
\newcommand{\tG}{\widetilde{G}}

\renewcommand{\Pr}{\mathbb{P}}
\newcommand{\eps}{\varepsilon}

\newcommand{\ab}[1]{\left<{#1}\right>} 
\newcommand{\rb}[1]{\left( #1 \right)} 

\newcommand{\vstar}{v^\star}
\newcommand{\Cstar}{C^\star}

\newcommand{\clusterw}{\textsc{ClusterW}}
\newcommand{\ppivot}{\textsc{PPivot}}
\newcommand{\ouralgo}{\textsc{OurAlgo}}
\newcommand{\clusterwparam}[1]{\textsc{ClusterW-#1}}
\newcommand{\ppivotparam}[1]{\textsc{PPivot-#1}}
\newcommand{\ouralgoparam}[1]{\textsc{OurAlgo-#1}}
\newcommand{\datadblp}{\textsf{dblp}}
\newcommand{\datauk}{\textsf{uk}}
\newcommand{\datait}{\textsf{it}}
\newcommand{\datatwitter}{\textsf{twitter}}
\newcommand{\datawebbase}{\textsf{webbase}}

\newcommand{\prob}[1]{\Pr \left[ #1 \right]}

\newcommand{\E}[1]{{\mathbb{E}}\left[#1\right]}

\newcommand{\eqdef}{\stackrel{\text{\tiny\rm def}}{=}}

\begin{document}

\onecolumn

\icmltitle{Correlation Clustering in Constant Many Parallel Rounds}




\begin{icmlauthorlist}
\icmlauthor{Vincent Cohen-Addad}{goo}
\icmlauthor{Silvio Lattanzi}{goo}
\icmlauthor{Slobodan Mitrović}{mit}
\icmlauthor{Ashkan Norouzi-Fard}{goo}
\icmlauthor{Nikos Parotsidis}{goo}
\icmlauthor{Jakub Tarnawski}{msr}
\end{icmlauthorlist}

\icmlaffiliation{mit}{CSAIL, MIT, Cambridge, MA, USA}
\icmlaffiliation{goo}{Google Research, Zürich, Switzerland}
\icmlaffiliation{msr}{Microsoft Research, Redmond, WA, USA}

\icmlkeywords{Correlation clustering, MPC, approximation algorithm}

\vskip 0.3in

\printAffiliationsAndNotice{}

\begin{abstract}
Correlation clustering is a central topic in unsupervised learning, with many applications in ML and data mining.  In correlation clustering, one receives as input a signed graph and the goal is to partition it to minimize the number of disagreements. In this work we propose a massively parallel computation (MPC) algorithm for this problem that is considerably faster than prior work. In particular, our algorithm uses machines with memory sublinear in the number of nodes in the graph and returns a constant approximation while running only for a constant number of rounds. To the best of our knowledge, our algorithm is the first that can provably approximate a clustering problem on graphs using only a constant number of MPC rounds in the sublinear memory regime. We complement our analysis with an experimental analysis of our techniques.
\end{abstract}

\section{Introduction}

Clustering is a classic problem in machine learning. The goal of clustering is to partition a given set of objects into sets
so that objects in the same cluster are similar to each other while objects in different clusters are dissimilar. One of the
most studied formulations of this problem is \emph{correlation clustering}.
Thanks to its simple and natural formulation,
this clustering variant has many applications in finding clustering ensembles~\cite{bonchi2013overlapping}, in duplicate 
detection~\cite{arasu2009large}, community detection~\cite{chen2012clustering}, disambiguation 
tasks~\cite{kalashnikov2008web}, and automated labelling~\cite{agrawal2009generating, chakrabarti2008graph}.

Correlation clustering was first formulated by~\citet{bansal2004correlation}.
Formally, in this problem we are given as input a weighted graph with $n$ nodes,
where positive edges represent similarities between nodes and negative edges represent dissimilarities
between them. We are interested in clustering the nodes to minimize the sum of the weights of the negative edges
contained inside any cluster plus the sum of positive edges crossing any two clusters. The problem is known to be NP-hard, and much
attention has been paid to designing approximation algorithms for the minimization version of the problem, as well as for its
complementary version where one is interested in maximizing agreement. In particular, for the most studied version of the
problem, where the weights are restricted to be in $\{-1,+1\}$, a polynomial-time approximation scheme is known for
the maximization version of the problem~\cite{bansal2004correlation} and a $2.06$-approximation algorithm is known for its
minimization version~\cite{chawla2015near}. Furthermore, when weights are in $\{-1,+1\}$ and the number of clusters is
upper-bounded by $k$, a polynomial-time approximation scheme is known also for the minimization version of the 
problem~\cite{giotis2005correlation}. For arbitrary weights, we know a $0.7666$-approximation algorithm for the maximization
version of the problem~\cite{charikar2005clustering, swamy2004correlation} and an $O(\log n)$-approximation for the 
minimization version of the problem~\cite{demaine2006correlation}.

One main drawback of classic solutions for correlation clustering is that they do not scale very well to very large networks. 
Thus, as the magnitude of available data grows, it becomes increasingly important to design efficient parallel algorithms for 
this problem. Unfortunately, obtaining such algorithms is often challenging because classic solutions to graph problems are
inherently sequential, e.g., the algorithm is defined iteratively and in an adaptive manner. Concretely, a well-known and widely used
algorithm for the unweighted minimization version of the problem requires solving a linear program or running a so-called \emph{Pivot}
algorithm \cite{ailon2008aggregating,chawla2015near}.
Designing an efficient parallel linear program solver, if one exists at all, is a major challenge. 
The Pivot algorithm is extremely elegant
and simple: it starts by selecting a node uniformly at random in the graph; then it creates a cluster by clustering together the 
node with all its positive neighbors; finally the algorithm recurs on the rest of the graph. Interestingly, this simple algorithm
returns a $3$-approximation to the minimization version of the problem when the weights are in $\{-1,+1\}$.
However, despite its simplicity, it is quite challenging to parallelize this algorithm efficiently. A strong step in this direction was presented
by~\citet{chierichetti2014correlation}, who show how to approximately parallelize the Pivot algorithm using 
$O\left(\frac{\log^2 n}{\epsilon}\right)$ parallel rounds to obtain a $(3+\epsilon)-$approximation for the problem.
In a subsequent work, \citet{ahn2015correlation} present a nice result for the semi-streaming setting, which can be
adapted to provide a $3$-approximation by running
  $O\left(\log\log n\right)$ rounds and using $\tilde{O}(n)$ memory per machine.
	In another related work, \citet{pan2015parallel} propose a new algorithm that runs in 
  $O\left(\frac{\log n\log \Delta}{\epsilon}\right)$ rounds (where $\Delta$ is the maximum positive degree) and obtain
very nice experimental results. In a recent work, when the memory per machine is $o(n)$, \cite{cambus2021massively} show how to construct a $3$-approximate (in expectation) correlation clustering in graphs of arboricity $\lambda$ in $O(\log{\lambda} \cdot \poly(\log \log n))$ rounds. A natural important question has thus been: \emph{Is it possible to approximate unweighted
minimum disagreement in $o(\log n)$ many rounds with $o(n)$ memory per machine?}
In this paper we answer this question affirmatively. Moreover, we design an algorithm that requires only $O(1)$ rounds, thus improving on the existing approaches in regimes of both $\tilde{O}(n)$ and $o(n)$ memory per machine.
Next, we discuss the precise model of parallelism that we use in this work.

\textbf{The MPC model.} We design algorithms for the massively parallel computation (MPC) model, 
which is a theoretical abstraction of real-world parallel systems such as MapReduce \cite{dean2008mapreduce}, 
Hadoop \cite{white2012hadoop}, Spark \cite{zaharia2010spark} and Dryad \cite{isard2007dryad}. The MPC model~\cite{karloff2010model, goodrich2011sorting, beame2013communication} is widely used as the de-facto standard theoretical model 
for large-scale parallel computing.

In the MPC model, computation proceeds in synchronous parallel \emph{rounds} over multiple machines. Each machine has memory $S$.
At the beginning of the computation, data is arbitrarily partitioned across the machines. During each round, machines process data locally. 
At the end of a round, machines exchange messages, with the restriction that each machine is allowed to send messages and 
receive messages of total size $S$. The efficiency of an algorithm in this model is measured by the number of rounds it takes for 
the algorithm to terminate and by the size $S$ of the memory of every machine. In this paper we focus on the most practical and challenging
regime, also known as the \emph{sublinear} regime, where each machine has memory $S = O(n^\delta)$ where $\delta$ is an arbitrary constant
smaller than $1$.

\noindent\textbf{Our contribution.} Our main contribution is to present a constant-factor approximation algorithm
for the minimization problem
when the weights are in $\{-1,+1\}$. Our new algorithm runs using only a constant number of rounds in the sublinear regime.

\begin{restatable}{thr}{mpcmain}
\label{theorem:MPC-main}
	For any constant $\delta > 0$, there exists an MPC algorithm that, given a signed graph $G = (V, E^+)$, where $E^+$ denotes the set of edges with weight +1, in $O(1)$ rounds computes a $O(1)$-approximate correlation clustering. Letting $n = |V|$, this algorithm succeeds with probability at least $1-1/n$ and requires $O(n^{\delta})$ memory per machine. Moreover, the algorithm uses a total memory of $O(|E^+| \cdot \log{n})$.
\end{restatable}

To the best of our knowledge, this is the first MPC graph clustering algorithm that runs in a constant number of rounds in the sublinear regime.
Furthermore, we also show that our algorithms extend to the semi-streaming setting. In particular, in this setting, our algorithm outputs an $O(1)$-approximate correlation clustering in only $O(1)$ passes over the stream. In terms of the number of passes, this significantly improves on the $3$-approximate algorithm by \cite{ahn2015correlation}, which requires $O(\log \log{n})$ passes.

\begin{theorem}\label{theorem:semi-streaming-main}
  There exists a semi-streaming algorithm that, given a signed graph $G = (V, E^+)$, where $E^+$ denotes the set of edges with weight +1,
   in $O(1)$ passes computes a $O(1)$-approximate correlation clustering.
  Letting $n = |V|$, this algorithm succeeds with probability at least $1-1/n$.
\end{theorem}

We  complement our theoretical results with an empirical analysis showing that our MPC algorithm is significantly faster than 
previously known algorithms~\cite{chierichetti2014correlation, pan2015parallel}. Furthermore, despite its theoretical
approximation 
guarantees being inferior to previous work, in our experiments the quality of the solution is better.
We explain this as follows: (1) all clusters returned by our algorithm are guaranteed to be very dense,
as opposed to the pivot-based algorithms, where formed clusters might be sparse, and (2) the 
similarities between our algorithm and existing heuristics for clustering that are known to work very
well in practice~\cite{xu2007scan}. 

\noindent\textbf{Techniques and Roadmap.} In contrast to previous known parallel
algorithms~\cite{chierichetti2014correlation,ahn2015correlation, pan2015parallel}, our algorithm is not based on a
parallel adaptation of the Pivot algorithm. Instead, we
study structural properties of correlation clustering. We show that, up to losing a constant factor in the quality
of the solution, one can simply focus on clusters consisting of points whose neighborhoods are almost identical
(up to a small multiplicative factor); we call such points ``in agreement'' and focus on clusters of such points. The next key idea is to trim the input graph so as to
only keep edges between some specific points in agreement (intuitively, the
points that are in agreement with many of their neighbors), so that clusters of points in agreement correspond to
connected components in this trimmed graph. We show how all the above operations can be performed in few rounds.
Finally, it remains to simply compute the connected components of the trimmed graph to obtain the final clusters.
Here we prove an important feature of the trimmed graph: Each connected component has constant diameter. This
ensures that this last step can indeed be performed in few parallel rounds.

In \cref{sec:algorithm} we present our algorithm, then in \cref{section:analysis} we present its analysis and in \cref{section:implementation} its MPC implementation. In \cref{section:streaming-implementation} we show how to extend these ideas to the semi-streaming setting.
Finally we present our experimental results in \cref{section:experiments}.

\section{Preliminaries}
\label{sec:preliminaries}
In this paper, we study the \emph{min-disagree} variant of correlation clustering in
the ``complete graph case'', where we are given a complete graph and
each edge is labeled with either $-1$ or $+1$: $E^+$ and $E^-$ respectively denote the set of
 edges labeled $+1$ and $-1$. The goal is to find a partition $C_1, \ldots, C_t$ of the vertices of the graph that minimizes the following objective:
\[
	f(C_1, \ldots, C_t) = \sum_{\substack{\{u, v\} \in E^+ : \\ u \in C_i, v \in C_j, i \neq j}} 1 + \sum_{\substack{\{u, v\} \in E^- : \\ u, v \in C_i}} 1.
\]

\paragraph{Notation.}
In our analysis and discussion, instead of working with a signed graph $(V, E^+, E^-)$, we work with an unweighted undirected graph $G = (V, E)$, where $E$ refers to $E^+$. Therefore, by this convention, $E^- = {V \choose 2} \setminus E$. In addition, when we say that two vertices $u$ and $v$ are neighbors, we mean that $\{u, v\} \in E = E^+$.

We use some standard notation that we briefly recall here. For a vertex $v\in V$, we refer to its neighborhood by $N(v)$ and to its degree by $d(v)$; we further let $N(v, H)$ denote the  neighborhood of $v$ in a subgraph $H$ of $G$. We also consider the degree of a vertex $v$ induced on an arbitrary subset $S\subseteq V$ of nodes and we denote it by $d(v, S)$. We also refer to the hop-distance between two vertices $u,v\in V$ by $\dist^G(u,v)$. We consider the hop distance also in subgraphs $\tG$ of $G$, in which case we denote the distance by $\dist^{\tG}(u,v)$. Finally for any two sets $R, S$, we denote their symmetric difference by $R \triangle S$.

\begin{remark}
    We assume that each vertex has a self-loop ``+'' edge. Note that this does not affect the cost of clustering, as a self-loop is never cut by a clustering. Note that this assumption implies that $v \in N(v)$.
\end{remark}

\section{Algorithm}\label{sec:algorithm}
The starting point of our approach is the notion of \emph{agreement} between vertices. Informally, we say that $u$ and $v$ are in agreement when their neighborhoods significantly overlap. Intuitively, in such scenario, we expect $u$ and $v$ to be treated equally by an algorithm: either $u$ and $v$ are in the same cluster, or both of them form singleton clusters.

Our algorithms are parametrized by two constants $\beta$, $\lambda$ that will be determined later.
\begin{definition}[Weak Agreement]\label{definition:agreement}
    Two vertices $u$ and $v$ are in \emph{$i$-weak agreement} if $|N(u) \triangle N(v)| < i \beta \cdot \max\{|N(u)|, |N(v)|\}$. If $u$ and $v$ are in $1$-weak agreement, we also say that $u$ and $v$ are in \emph{agreement}.
\end{definition}
Having the agreement notion in hand, we provide our approach in \cref{alg:main}.
\begin{algorithm}[h]
\begin{algorithmic}[1]
\caption{Correlation-Clustering($G$) \label{alg:main}}
	\STATE {Discard all edges whose endpoints are not in agreement. (First compute the set of these edges. Then remove this set.) \alglinelabel{step:not-in-agreement}}
	
	\STATE Call a vertex \emph{light} if it has lost more than a $\lambda$-fraction of its neighbors in the previous step. Otherwise call it \emph{heavy}. \alglinelabel{line:heavy-and-light-vertices}
	
	\STATE {Discard all edges between two light vertices. \alglinelabel{step:discard-light-light}}
	
	\STATE Call the current graph $\tG$, or the \emph{sparsified graph}. Compute its connected components, and output them as the solution. \alglinelabel{step:connected-components-in-tG}
\end{algorithmic}
\end{algorithm}

\subsection{Analysis}
\label{section:analysis}
Our analysis consists of two main parts. The first part consists of analyzing properties of $\tG$ and, in particular, showing that each connected component of $\tG$ has $O(1)$ diameter, with all vertices being in $O(1)$-weak agreement. The second part shows that the number of edges removed in Lines~\ref{step:not-in-agreement}-\ref{step:discard-light-light} of \cref{alg:main} is only a constant factor larger than the cost an optimal solution.

\subsubsection{Properties of $\tG$}
Our analysis hinges on several properties of vertices being in weak agreement. We start by stating those properties.
\begin{fact}\label{main:fact}
Suppose that $\beta < \frac{1}{20}$.
\begin{enumerate}[(1)]
	\item \label{fact:degree-bounds}     If $u$ and $v$ are in $i$-weak agreement, for some $1 \le i < \frac{1}{\beta}$, then
    \[
        (1 - \beta i) d(u) \le d(v) \le \frac{d(u)}{1 - i\beta}.
    \]
	\item \label{fact:triangle_inequality}
Let $k \in \{2,3,4,5\}$ and $v_1, \ldots, v_k \in V$ be a sequence of vertices such that $v_i$ is in agreement with $v_{i+1}$ for $i=1, \ldots,k-1$. Then $v_1$ and $v_k$ are in $k$-weak agreement.
	\item \label{fact:intersection-bound}
    If $u$ and $v$ are in $i$-weak agreement, for some $1 \le i < \frac{1}{\beta}$, then
    $|N(v) \cap N(u)| \ge (1 - i \beta) d(v)$.
\end{enumerate}
\end{fact}
\begin{proof}
\begin{enumerate}[(1)]
\item
    Without loss of generality, assume that $d(u) \le d(v)$.
    We have $|N(u) \triangle N(v)| \ge d(v) - d(u)$. Then, by \cref{definition:agreement}, 
   $d(v) - d(u) \le |N(u) \triangle N(v)| \le i \beta \cdot d(v)$.
    This now implies $d(u) \ge (1 - i \beta) d(v)$, as desired.

\item
	For $i=1,...,k-1$ we have by \eqref{fact:degree-bounds}:
	\[
	d(v_i) \le \frac{d(v_{i+1})}{1 - \beta} \le ... \le \frac{d(v_k)}{(1 - \beta)^{k-i}} \le \frac{d(v_k)}{(1 - \beta)^4} \le \frac{k}{k-1} \cdot d(v_k) \,, \]
	since $(1 - \beta)^4 \ge (1 - \frac{1}{20})^4 > \frac{4}{5} \ge \frac{k-1}{k}$.
	Now we iterate the triangle inequality:
	\begin{align*}
	|N(v_1) \triangle N(v_k)| &\le \sum_{i=1}^{k-1} |N(v_i) \triangle N(v_{i+1})| \\
	&< \sum_{i=1}^{k-1} \beta \cdot \max(d(v_i), d(v_{i+1})) \\
	&\le (k-1) \cdot \beta \cdot \frac{k}{k-1} \cdot d(v_k) \\
	&\le k \cdot \beta \cdot \max(d(v_1), d(v_k)) \,.
	\end{align*}

\item
    Without loss of generality, assume that $d(u) \le d(v)$.
    Then
    \begin{align*}
			|N(u) \cap N(v)| & = |N(v)| - |N(v) \setminus N(u)| \\
			& \ge |N(v)| - |N(u) \triangle N(v)| \\
			& \ge (1 - i \beta) d(v). 
    \end{align*}
\end{enumerate}
\end{proof}

By building on these claims, we are able to show that $\tG$ has a very convenient structure: each of its connected components has diameter of only at most $4$; and every two vertices (one of them being heavy) in a connected component of $\tG$ are in $4$-weak agreement. More formally, we have:

\begin{restatable}{lma}{diameter}
\label{lemma:diameter}
	Suppose that $5 \beta + 2 \lambda < 1$.
    Let $CC$ be a connected component of $\tG$.
    Then, for every $u,v \in CC$:
		\vspace{-10pt}
    \begin{enumerate}[(a)]
        \item if $u$ and $v$ are heavy, then $\dist^{\tG}(u,v) \le 2$,
				\vspace{-5pt}
        \item $\dist^{\tG}(u,v) \le 4$,
				\vspace{-5pt}
        \item $\dist^G(u,v) \le 2$,
				\vspace{-5pt}
        \item\label{item:2-weak-agreement} if $u$ or $v$ is heavy, then $u$ and $v$ are in $4$-weak agreement.
				\vspace{-5pt}
    \end{enumerate}
\end{restatable}
\begin{proof}
    For (a),    suppose by contradiction that there are heavy $u,v \in CC$ with $\dist^{\tG}(u,v) > 2$; pick such $u,v$ with minimum $\dist^{\tG}(u,v)$.
    If $\dist^{\tG}(u,v) \ge 5$, let $P = \ab{u,u',u'',...,v}$ be a shortest $u$-$v$ path in $\tG$; since there are no edges in $\tG$ with both endpoints being light, either $u'$ or $u''$ must be heavy, and the pair $(u',v)$ or $(u'',v)$ contradicts the minimality of the path $(u,v)$ (as we have $\dist^{\tG}(u'',v) > 2$).

    On the other hand, if $\dist^{\tG}(u,v) \le 4$,
    then by \cref{main:fact}~\eqref{fact:triangle_inequality} $u$ and $v$ are in $5$-weak agreement,
    and by \cref{main:fact}~\eqref{fact:intersection-bound} we have $|N(u) \cap N(v)| \ge (1-5\beta)d(v)$.
    Note that a heavy vertex can lose at most a $\lambda$ fraction of its neighbors in $G$ in Line~\ref{step:not-in-agreement} of the algorithm, and it loses no neighbors in Line~\ref{step:discard-light-light}; thus $|N(v) \setminus N^{\tG}(v)| \le \lambda d(v)$ and similarly for $u$.
    Assume without loss of generality that $d(v) \ge d(u)$.
    Then we have
    \[
    |N^{\tG}(u) \cap N^{\tG}(v)| \ge |N(u) \cap N(v)| - |N(u) \setminus N^{\tG}(u)| - |N(v) \setminus N^{\tG}(v)| \ge (1-5 \beta - 2 \lambda) d(v) > 0 \,,
    \]
    i.e., $u$ and $v$ have a common neighbor in $\tG$, and thus, $\dist^{\tG}(u,v) \le 2$.
    
    For (b), let $P$ be a shortest $u$-$v$ path in $\tG$.
    Define the vertex $u'$ to be $u$ if $u$ is heavy and to be $u$'s neighbor on $P$ if $u$ is light; in the latter case, $u'$ is heavy since there are no edges in $\tG$ with both endpoints being light.
    Define $v'$ similarly.
    Since $u'$ and $v'$ are heavy, we have $\dist^{\tG}(u,v) \le 1 + \dist^{\tG}(u',v') + 1 \le 4$.
    
    For (c), note that by (b) and \cref{main:fact}~\eqref{fact:triangle_inequality}, $u$ and $v$ are in $5$-weak agreement; by \cref{main:fact}~\eqref{fact:intersection-bound}, they have at least $(1-5\beta)d(v) > 0$ common neighbors in $G$.
    
    To prove (d), we proceed similarly as for (b). We consider two cases: both $u$ and $v$ are heavy; only one $u$ or $v$ is heavy. In the first case, by (a) and \cref{main:fact}~\eqref{fact:triangle_inequality} we even have that $u$ and $v$ are in $3$-weak agreement. In the second case, one of the vertices is light; without loss of generality, assume $u$ is light. In that case, $u$ is adjacent to a heavy vertex $u'$, as there are no edges between light vertices. Since by (a) $v$ and $u'$ are at distance $2$, it implies that $v$ and $u$ are at distance $3$. Since each edge $(x,y)$ in $CC$ means that $x$ and $y$ are in agreement, by \cref{main:fact}~\eqref{fact:triangle_inequality} we have that $v$ and $u$ are in $4$-weak agreement.
\end{proof}

We now illustrate how to apply \cref{lemma:diameter} to show further helpful properties of connected components of $\tG$. First, observe that a non-trivial connected component $CC$ of $\tG$ (i.e., one consisting of at least two vertices) has at least one heavy vertex. (As a reminder, heavy vertices are defined on Line~\ref{line:heavy-and-light-vertices} of \cref{alg:main}.) Indeed, any edge in $\tG$ has at least one heavy endpoint, as assured by Line~\ref{step:discard-light-light} of \cref{alg:main}. Let $x$ be such a heavy vertex. Then, by Property~\eqref{item:2-weak-agreement} of \cref{lemma:diameter} we have that \emph{every} other vertex in $CC$ shares a large number of neighbors with $x$. One can turn this property into a claim stating that all vertices in $CC$ have induced degree inside $CC$  very close to $|CC|$. Formally:

\begin{restatable}{lma}{clusterdegreelowebound}\label{lemma:lower-bound-in-cluster-degree}
    Let $CC$ be a connected component of $\tG$ such that $|CC| \ge 2$. Then, for each vertex $u \in CC$ we have that
    \[
        d(u, CC) \ge (1 - 8 \beta - \lambda) |CC|.
    \]
\end{restatable}

(Note that $d(u, CC)$ in \cref{lemma:lower-bound-in-cluster-degree} is defined with respect to the edges appearing in $G$.)

\begin{proof}
    Assume that $CC$ is a non-trivial connected component, i.e., $CC$ has at least two vertices.
    Let $x$ be a heavy vertex in $CC$. Observe that such a vertex $x$ always exists by the construction of our algorithm -- edges having both light endpoints are removed in Line~\ref{step:discard-light-light} of \cref{alg:main}.
    
    \emph{Remark:} While $CC$ refers to a connected component in the sparsified graph $\tG$, note that $N(\cdot)$ and $d(\cdot)$ refer to neighborhood and degree functions with respect to the input graph $G$ rather than with respect to $\tG$.

    First, from \cref{lemma:diameter}~\eqref{item:2-weak-agreement}, we have that any two vertices in $CC$, one of which is heavy, are in $4$-weak agreement. In particular, this also holds for $x$ and any other vertex $u\in CC$.
    As defined in \cref{sec:preliminaries}, recall that $N(x, CC) \eqdef N(x) \cap CC$. Since $x$ is a heavy vertex, it has at most a $\lambda$-fraction of its neighbors $N(x)$ outside $CC$, and so from \cref{main:fact}~\eqref{fact:intersection-bound} we have
    \begin{equation}\label{eq:lower-bound-intersection-in-CC}
        |N(x, CC) \cap N(u)| \ge (1 - 4 \beta) d(x) - \lambda d(x) = (1 - 4 \beta - \lambda) d(x).
    \end{equation}
    Observe that this also implies
    \begin{equation}\label{eq:-lower-bound-in-CC-degree}
        |N(u, CC)| \ge (1 - 4 \beta - \lambda) d(x).
    \end{equation}
    
    Next, we want to upper-bound the number of vertices in $CC \setminus N(x)$, which will enable us to express $|CC|$ as a function of $d(x)$. To that end, note that 
    \cref{eq:lower-bound-intersection-in-CC} implies a lower bound on the number of edges between the neighbors of $x$ in $CC$, denoted by $N(x, CC)$, and the vertices in $CC$ other than $N(x)$, denoted by $CC \setminus N(x)$, as follows:
    \begin{equation}\label{eq:lower-bound-edges}
        |E(N(x, CC), CC \setminus N(x))| \ge |CC \setminus N(x)| \cdot (1 - 4 \beta - \lambda) d(x),
    \end{equation}
    where $E(Y, Z)$ is the set of edges between sets $Y$ and $Z$.
    On the other hand, since $d(u) \le \frac{d(x)}{1 - 4 \beta}$ for each $u \in CC$ by \cref{main:fact}~\eqref{fact:degree-bounds} and since $u$ and $x$ are in $4$-weak agreement, we have that $u$ has at most $4 \beta \frac{d(x)}{1 - 4 \beta}$ neighbors outside $N(x)$. Hence, we derive
    \[
        |E(N(x, CC), CC \setminus N(x))| \le |N(x, CC)| \cdot \frac{4 \beta d(x)}{1 - 4 \beta} \le d(x) \cdot \frac{4 \beta d(x)}{1 - 4 \beta}.
    \]
    Combining the last inequality with \cref{eq:lower-bound-edges} yields
    \[
        |CC \setminus N(x))| \le \frac{4 \beta d(x)}{(1 - 4 \beta) \cdot (1 - 4 \beta - \lambda)} \le \frac{4 \beta d(x)}{1 - 8 \beta - \lambda},
    \]
    which further implies
    \[
        |CC| = |CC \setminus N(x)| + |N(x, CC)| \le \rb{1 + \frac{4 \beta }{1 - 8 \beta - \lambda}} d(x) = \frac{1 - 4 \beta - \lambda}{1 - 8 \beta - \lambda} d(x).
    \]
    Now together with \cref{eq:-lower-bound-in-CC-degree}, we establish
    \[
        |N(u, CC)| \ge (1 - 8 \beta - \lambda) |CC|,
    \]
    as desired.
\end{proof}

Building on \cref{lemma:lower-bound-in-cluster-degree} we can now show that it is not beneficial to split a connected component into smaller clusters. Intuitively, this is the case as each vertex in a connected component $CC$ has degree almost $|CC|$, while splitting $CC$ into at least two clusters would force the smallest cluster (that has size at most $|CC|/2$) to cut too many ``+'' edges, while in $CC$ it has relatively few ``-'' edges.

\begin{restatable}{lma}{keepconnectedcomponent}\label{lemma:keep-connected-components}
    Let $CC$ be a connected component in $\tG$. Assume that $8 \beta + \lambda \le 1/4$. Then, the cost of keeping $CC$ as a cluster in $G$ is no larger than the cost of splitting $CC$ into two or more clusters.
\end{restatable}

\begin{proof}
    Towards a contradiction, consider a split of $CC$ into $k \ge 2$ clusters $C_1, \ldots, C_k$ whose cost is less than the cost of keeping $CC$ as a single cluster.
    Moreover, consider the cheapest such split of $CC$. Let $\delta \eqdef 8 \beta + \lambda$. We consider two cases: when each cluster in  $\{C_1, \ldots, C_k\}$ has size at most $(1 - 2 \delta) |CC|$ vertices, and the complement case.
    
    \paragraph{It holds that $|C_i| \le (1 - 2 \delta) |CC|$ for each $i$.}
        By \cref{lemma:lower-bound-in-cluster-degree}, each vertex $v \in C_i$ for each cluster $C_i$ has at least $(1 - \delta)|CC| - |C_i| \ge \delta |CC|$ neighbors in $CC \setminus C_i$. Hence, splitting $CC$ in the described way cuts at least $\frac{\delta |CC|^2}{2}$ ``+'' edges. On the other hand, also by \cref{lemma:lower-bound-in-cluster-degree}, $CC$ has at most $\frac{\delta |CC|^2}{2}$ ``-'' edges. Hence, it does not cost less to split $CC$ in the described way.
    
    \paragraph{There exists a cluster $\Cstar$ such that $|\Cstar| > (1 - 2 \delta) |CC|$.}
    Let $C_i \neq \Cstar$ be one of the clusters $CC$ is split into. Clearly, we have $|C_i| < 2 \delta |CC|$.
    Since, by \cref{lemma:lower-bound-in-cluster-degree}, each vertex $v \in C_i$ has at least $(1 - \delta) |CC|$ ``+'' edges inside $CC$, it implies that $v$ has more than $(1 - 3 \delta) |CC|$ ``+'' edges to $\Cstar$. On the other hand, there are at most $\delta |CC|$ ``-'' edges from $v$ to $\Cstar$. Hence, as long as $1 - 3 \delta \ge \delta$, it implies that it is \emph{cheaper} to merge $\Cstar$ with $C_i$ than to keep them split. This contradicts our assumption that the split into those $k$ clusters results in the minimum cost.
    
    Observe that the condition $1 - 3 \delta \ge \delta$ is equivalent to $8 \beta + \lambda \le 1/4$, which holds by our assumption.
\end{proof}

\cref{lemma:keep-connected-components} implies the following key insight.

\begin{restatable}{lma}{optfortg}
  \label{lemma:opt-for-tg}
    Let $G'$ be a \emph{non-complete}\footnote{We remark that everywhere else in the paper, correlation clustering instances are always complete graphs.} graph obtained from $G$ by removing any ``+'' edge $\{u, v\}$ (i.e., changing it into a ``neutral'' edge) where $u$ and $v$ belong to different connected components of $\tG$.
    Then, our algorithm outputs a solution that is optimal for the instance $G'$.
\end{restatable}

\begin{proof}
    It is suboptimal for a single cluster to contain vertices from different connected components;
    indeed, breaking such a cluster up into connected components would improve the objective function (all edges between connected components are negative).
    Therefore any optimal solution must either be equal to our solution or it should split some cluster in our solution.
    The claim follows, by \cref{lemma:keep-connected-components}, because subdividing a connected component of $G'$ (equivalently of $\tG$) does not improve the objective function. 
\end{proof}

\subsubsection{Approximation Guarantee} \label{sec:approximation}

In our analysis we will consider a fixed 
optimal solution (of instance $G$), denoted by $\cO$, whose cost is denoted by $\OPT$.

Recall that our algorithm returns a clustering that is optimal for $G'$ (\cref{lemma:opt-for-tg}).
Therefore to bound the approximation ratio of our solution we need to bound the cost in $G$
of an optimal clustering for $G'$. To do so, it is enough to bound the number of
``+'' edges in $G$ that are absent from $G'$ -- and every such edge has been deleted by our algorithm.
We have the following two lemmas. The main intuition behind their proofs
is that when two vertices are not in agreement,
or when a vertex is light, then there are many edges (or non-edges) in the 1-hop or 2-hop vicinity
that $\cO$ pays for. We can charge the deleted edges to them.

\begin{restatable}{lma}{removededgesstepone} \label{lemma:cut-step-1}
    The number of edges deleted in Line~\ref{step:not-in-agreement} of our algorithm
    that are not cut in $\cO$
    is at most $\frac{2}{\beta} \cdot \OPT$.
\end{restatable}
\begin{proof}
    Our proof is based on a charging argument. Each edge as in the statement will distribute fractional debt to edges (or non-edges) that $\cO$ pays for, in such a way that (1) each edge as in the statement distributes debt worth at least $1$ unit, and (2) each edge/non-edge that $\cO$ pays for is assigned at most $\frac{2}{\beta}$ units of debt, (3)
    edges/non-edges that $\cO$ does \emph{not} pay for are  assigned no debt.
    
    Let $(u,v)$ be an edge as in the statement
    (its endpoints are not in agreement).
    That is, we have $|N(u) \triangle N(v)| > \beta \cdot \max(d(u),d(v))$,
    and $u$, $v$ belong to the same cluster in $\cO$.
    Then, for each $w \in |N(u) \triangle N(v)|$,
    $\cO$ pays for one of the edges/non-edges $(u,w)$, $(v,w)$.
    (If $w$ is in the same cluster as $u,v$, then $\cO$ pays for the one of $(u,w)$, $(v,w)$ that is a non-edge; and vice versa).
    So $(u,v)$ can assign $\frac{1}{\beta \cdot \max(d(u),d(v))}$ units of debt to that edge/non-edge.
    This way, properties (1) and (3) are clear.
    
    We verify property (2).
    Fix an edge/non-edge $(a,b)$ that $\cO$ pays for.
    It is only charged by adjacent edges.
    Each edge adjacent to $a$, of which there are $d(a)$ many,
    assigns at most $\frac{1}{\beta \cdot d(a)}$ units of debt;
    this gives $\frac{1}{\beta}$ units in total.
    The same holds for edges adjacent to $b$;
    together this yields $\frac{2}{\beta}$ units.
\end{proof}

\begin{restatable}{lma}{removededgesstepthree} \label{lemma:cut-step-3}
    The number of edges deleted in Line~\ref{step:discard-light-light} of our algorithm
    that are not cut in $\cO$
    is at most $\rb{\frac{1}{\beta} + \frac{1}{\lambda} + \frac{1}{\beta \lambda}} \cdot \OPT$.
\end{restatable}

\begin{proof}
    We use a similar charging argument as in the proof of \cref{lemma:cut-step-1},
    with the difference that each edge/non-edge that $\cO$ pays for
    will be assigned at most $\frac{1}{\beta} + \frac{1}{\lambda} + \frac{1}{\beta \lambda}$ units of debt (rather than at most $\frac{2}{\beta}$).
    
    Let $(u,v)$ be an edge as in the statement.
    For each endpoint $y \in \{u,v\}$, we proceed as follows.
    As $y$ is light,
    there are edges $(y,v_1), ..., (y,v_{\lambda \cdot d(y)})$
    whose endpoints are not in agreement.
    For each $i = 1, ..., \lambda \cdot d(y)$, proceed as follows:
    \begin{itemize}
        \item If $(y,v_i)$ is not cut by $\cO$,
        then,
        as in the proof of \cref{lemma:cut-step-1},
        $(y,v_i)$
        has at least $\beta \cdot \max(d(y),d(v_i))$ adjacent edges/non-edges
        for whom $\cO$ pays.
        Each of these edges/non-edges
        is of the form $(v_i,w)$ or $(y,w)$.
        We will have the edge $(u,v)$
        charge $\frac{1}{2 \beta \lambda d(v_i) d(y)}$
        units of debt,
        which we will call \textbf{blue debt},
        to the former ones (those of the form $(v_i,w)$),
        and
        $\frac{1}{2 \beta \lambda d(y)^2}$
        units of debt,
        which we will call \textbf{red debt},
        to the latter ones (those of the form $(y,w)$).\footnote{
            Notice that the latter edges/non-edges might be charged many times by the same $y$ (for different $i$).
        }
        
        \item If $(y,v_i)$ is cut by $\cO$,
        then $\cO$ pays for $(y,v_i)$.
        We will have the edge $(u,v)$
        charge $\frac{1}{2 \lambda d(y)}$
        units of debt,
        which we will call \textbf{green debt},
        to $(y,v_i)$.
    \end{itemize}
    Let us verify property (1).
    In the first case,
    each of these edges/non-edges
    is charged at least $\frac{1}{2 \beta \lambda d(y) \max(d(y), d(v_i))}$
    units of debt,
    and since there are at least $\beta \cdot \max(d(y),d(v_i))$ of them,
    the total (blue or red) debt charged is at least
    $\frac{1}{2 \lambda d(y)}$
    per each $y \in \{u,v\}$ and
    each $i = 1, ..., \lambda \cdot d(y)$.
    This much total (green) debt is also charged in the second case.
    Since there are $2$ choices for $y$ and then $\lambda \cdot d(y)$ choices for $i$,
    in total the edge $(u,v)$ assigns at least $1$ unit of debt.
    Property (3) is satisfied by design.
    
    We are left with verifying property (2).
    Fix an edge/non-edge $(a,b)$ that $\cO$ pays for.
    It can be charged by its adjacent edges (red or green debt), as well as those at distance two (blue debt).
    Let us consider these cases separately.
    
    \textbf{Adjacent edges (red/green debt):}
    let us first look at edges adjacent to $a$
    (we will get half of the final charge this way).
    That is, $a$ is serving the role of $y$ above;
    it can serve that role
    for at most $d(a)$ debt-charging edges (serving the role of $(u,v)$, where $a = y \in \{u,v\}$).
    \begin{itemize}
    \item 
        \textbf{Red debt}:
        each of these debt-charging edges charges $(a,b)$ at most $\lambda \cdot d(a)$ times (once per $i = 1, ..., \lambda \cdot d(y)$),
        and each charge is for
        $\frac{1}{2 \beta \lambda d(a)^2}$
        units of debt.
        This gives $\frac{1}{2\beta \lambda d(a)^2} \cdot \lambda d(a) \cdot d(a) = \frac{1}{2 \beta}$ units of debt.
    \item 
        \textbf{Green debt}:
        each of these debt-charging edges charges $(a,b)$ at most once
        (if it happens that $(a,b) = (y,v_i)$ for some $i$),
        and each charge is for
        $\frac{1}{2 \lambda d(a)}$
        units of debt.
        This gives $\frac{1}{2 \lambda}$ units of debt.
    \end{itemize}
    We get the same amount from edges adjacent to $b$
    ($b$ serving the role of $y$).
    In total,
    we get a debt of $\frac{1}{\beta} + \frac{1}{\lambda}$.

    \textbf{Blue debt}:
    $(a,b)$ is serving the role of $(v_i,w)$ above.
    Let us first look at $a$ serving the role of $v_i$
    (we will get half of the final charge this way).
    Then a neighbor of $a$ must be serving the role of $y$.
    There are at most $d(a)$ possible $y$'s,
    and at most $d(y)$ possible edges $(u,v)$ for each $y$ (those with $y \in \{u,v\}$).
    Recall that each charge was for
    $\frac{1}{2 \beta \lambda d(v_i) d(y)} = \frac{1}{2 \beta \lambda d(a) d(y)}$
    units of debt;
    per $y$, this sums up (over edges $(u,v)$) to at most
    $\frac{1}{2 \beta \lambda d(a) d(y)} \cdot d(y) = \frac{1}{2 \beta \lambda d(a)}$
    total units,
    and since there are at most $d(a)$ many $y$'s,
    the total debt is at most $\frac{1}{2 \beta \lambda}$.
    We get the same amount from $b$ serving the role of $v_i$.
    In total,
    we get a debt of $\frac{1}{\beta \lambda}$.
\end{proof}

Lemmas \ref{lemma:opt-for-tg}, \ref{lemma:cut-step-1}, and \ref{lemma:cut-step-3} together imply that \cref{alg:main} is a constant-factor approximation:

\begin{restatable}{thr}{approximation} \label{theorem:approximation}
    \cref{alg:main} is a constant-factor approximation.
\end{restatable}
\begin{proof}
		Let $G'$ be the (non-complete) graph as defined in \cref{lemma:opt-for-tg}. Observe that the clusters that our algorithm outputs are exactly the connected components of $G'$.
		Let $D = E^+(G) \setminus E^+(G')$ be the set of edges in $G$ that go between different connected components of $G'$ (equivalently, of $\tG$).
	Further, recall that $\cO$ is a fixed optimal solution for instance $G$.

The main idea of our proof is to look at the costs of $\cO$ and of our solution in the instance $G'$, for which our solution is optimal.
	The cost of any solution differs between the two instances $G$ and $G'$ by at most $|D|$,
	which is at most the number of edges deleted by our algorithm.
	So, we can pay $|D|$ to move from $G'$ to $G$.
	On the other hand, any solution is no more expensive in $G'$ than it is in $G$.
	That is,
	for any solution $X$ we have 
	\[ \cost_{G'}(X) \le \cost_G(X) \le |D| + \cost_{G'}(X) \,. \]
	
	Denote the solution returned by \cref{alg:main} by $\OUR$.
	\cref{lemma:opt-for-tg} states that it
	is optimal 
	for $G'$.
	That is, $\cost_{G'}(\OUR) \le \cost_{G'}(\cO)$.
	Thus we have
	\begin{align*}
		\cost_G(\OUR) &\le |D| + \cost_{G'}(\OUR) \\&\le |D| + \cost_{G'}(\cO) \\&\le |D| + \cost_G(\cO) \\&= |D| + \OPT .
	\end{align*}
	
	
	Finally, note that $|D|$ is at most the number of edges deleted by our algorithm
	(since any edge of $G$ that goes between different connected components of $\tG$ must necessarily have been deleted by our algorithm).
	The latter can be upper-bounded, using Lemmas~\ref{lemma:cut-step-1} and~\ref{lemma:cut-step-3},
	by
	$\OPT + \frac{2}{\beta} \cdot \OPT + \rb{\frac{1}{\beta} + \frac{1}{\lambda} + \frac{1}{\beta \lambda}} \cdot \OPT$.
	In total,
	we get a $\rb{2 + \frac{3}{\beta} + \frac{1}{\lambda} + \frac{1}{\beta \lambda}}$-approximation.	
\end{proof}

We note that in our analysis we do not optimize for a constant; nevertheless we now present a
precise upper bound on the approximation ratio by providing a setting for the constants $\beta$ and $\lambda$. We also note that despite the large theoretical
approximation ratio, our algorithm works very well in practice.

Recall that \cref{lemma:diameter} requires that $5 \beta + 2 \lambda < 1$,
and \cref{lemma:keep-connected-components} requires $8 \beta + \lambda \le \frac{1}{4}$,
the latter condition being stronger.
Also, \cref{main:fact} requires $\beta < \frac{1}{20}$ (which is also implied by the above).
Thus we can set, e.g., $\beta = \lambda = \frac{1}{36}$.
Then the above proof of \cref{theorem:approximation} gives a
$1442$-approximation guarantee.
A more optimized setting of constants is $\beta \approx 0.0176$
and
$\lambda \approx 0.1085$,
which gives an approximation ratio $\approx 701$.

Finally, we have the following:
\begin{restatable}{rmrk}{analysisistight} \label{remark:analysis_is_tight}
    For fixed values of $\beta$ and $\lambda$, the above analysis is tight, in the sense that the term $\frac{1}{\beta \lambda}$ is necessary.
\end{restatable}

\begin{proof}
    Let us assume for simplicity that $\beta = \lambda$; otherwise the example can be adapted.
    Consider the following instance: two disjoint cliques $A_1$, $A_2$ of size $(1-\beta)d$ each,
    with a subset $X_1 \subseteq A_1$ and a subset $X_2 \subseteq A_2$,
    both of size $\beta d$, fully connected to each other.
    
    The optimal solution is to have two clusters ($A_1$ and $A_2$).
    The cost is $(\beta d)^2$
    (cutting the edges between $X_1$ and $X_2$).
    
    However, our algorithm will first delete the edges between $A_1 \setminus X_1$ and $X_1$ (any two vertices from these respective sets are not in agreement, as the $X_1$-vertex has $\beta d$ extra neighbors in $X_2$), between $X_1$ and $X_2$, and between $A_2 \setminus X_2$ and $X_2$.\footnote{
        As an aside, note that by now, the algorithm has paid around $\rb{1 + \frac{2}{\beta}} \cdot \OPT$, showing that \cref{lemma:cut-step-1} by itself is also tight for Line~\ref{step:not-in-agreement}.
    }
    Then every vertex in the graph becomes light. Thus in Line~\ref{step:discard-light-light} we delete all edges, making $\tG$ an empty graph.
    Finally, we return the singleton partitioning as the solution.
    Its cost is $(\beta d)^2 + 2 \cdot {{(d(1-\beta))^2} \choose 2} \approx \rb{\frac{1}{\beta^2} - \frac{2}{\beta} + 2} \cdot \OPT$.
\end{proof}

\subsection{MPC Implementation of \cref{alg:main}}
\label{section:implementation}
In this section we prove \cref{theorem:MPC-main}. The proof is divided into two parts: discussing the MPC implementation and proving the approximation ratio of the final algorithm.
There are two main steps that we need to implement in the MPC model: for each edge $\{u, v\}$, we need to compute whether $u$ and $v$ are in an agreement (needed for Line~\ref{step:not-in-agreement}); and to compute the connected components of $\tG$ (Line~\ref{step:connected-components-in-tG}). We separately describe how to implement these tasks.
The approximation analysis is given in \cref{section:MPC-approximation}.

\subsubsection{Computing Agreement}
\label{section:computing-agreement}
Let $e = \{u, v\}$ be an edge in $G$. To test whether $u$ and $v$ are in agreement, we need to compute how large $N(v) \triangle N(u)$ (or how large $N(u) \cap N(v)$) is (see \cref{definition:agreement}). However, it is not clear how to find $|N(v) \triangle N(u)|$ exactly for each edge $\{u, v\} \in E$ while using  total  memory of $\tO(|E|)$. So, instead, we will approximate $|N(v) \triangle N(u)|$ and use this approximation to decide whether $u$ and $v$ are in agreement. In particular, $u$ and $v$ will sample a small fraction of their neighbors, i.e., of size $O((\log{n}) / \beta)$, and then these samples will be used to approximate the similarity of their neighbourhoods. We now describe this procedure in more detail.

As the first step, we test whether $d(u)$ and $d(v)$ are within a factor $1-\beta$. If they are not, then by \cref{main:fact}~\eqref{fact:degree-bounds} $u$ and $v$ are not in agreement and hence we immediately remove the edge $\{u, v\}$ from $G$.
Next, each vertex $v$ creates two vertex-samples. To do so, for each $j$ smaller or equal than $O((\log{n}) / \beta)$ we define the set $S(j)$ as a subset of nodes obtained by sampling every node in the graph independently with probability $\min\left\{\frac{a \log{n}}{\beta \cdot j}, 1\right\}$, where $a$ is a constant to be fixed later. Then we define $S(v, j)$ for every node $v$ as $S(v, j) = S(j)\cap N(v)$ and $j_v$ to be the largest power of $1 / (1 - \beta)$ smaller or equal than $d(v)$. Then, each vertex $v$ keeps $S(v, j_v)$ and $S(v, j_v / (1 - \beta))$.
Note that by construction, for any two vertices $v$ and $u$, we either have that $w \in S(v, j)$ and $w \in S(u, j)$, or $w \notin S(v, j)$ and $w \notin S(u, j)$. To implement this, each vertex $w$ will independently in parallel flip a coin to decide whether for a given $j$ it should be sampled or not.

Once we obtain the two samples, $v$ sends the samples together with information about its degree to each of its incident edges.\footnote{We refer the reader to \cite{goodrich2011sorting} and Section~6 of \cite{czumaj2019round} for details on how to collect these samples on each edge in $O(1)$ MPC rounds.} After that, every edge $\{u, v\}$ holds: $S(v, j_v)$, $S(v, j_v / (1 - \beta))$, $S(u, j_u)$, and $S(u, j_u / (1 - \beta))$. Without loss of generality assume $d(u) \ge d(v)$. Since we have that $d(v) / (1 - \beta) \ge d(u) \ge d(v)$, then $j_v = j_u$ or $j_v / (1-\beta) = j_u$. For the sake of brevity, let $j = j_u$.
We now use $S(v, j)$ and $S(u, j)$ to estimate $|N(v) \triangle N(u)|$.

Define a random variable $X_{u, v}$ as 
\begin{equation}\label{eq:Xuv}
    X_{u, v} \eqdef |S(v, j) \triangle S(u, j)|.
\end{equation}
In case $\frac{a \cdot \log{n}}{\beta \cdot j} \ge 1$, we have $X_{u, v} = |N(v) \triangle N(u)|$, which means we directly get the \emph{exact} value of $|N(v) \triangle N(u)|$. So assume that $\frac{a \cdot \log{n}}{\beta \cdot j} < 1$.
By linearity of expectation we have
\[
    \E{X_{u, v}} = \frac{a \cdot \log{n}}{\beta \cdot j} |N(v) \triangle N(u)|.
\]

Hence, if $v$ and $u$ are in agreement, we have
\[
    \E{X_{u, v}} \le \frac{a \cdot \log{n}}{\beta \cdot j} \beta d(u)
    = \frac{a \cdot \log{n}}{j} d(u).
\]

Based on this, our algorithm for deciding whether $u$ and $v$ are in agreement is given as \cref{alg:agreement-sampling}.

\begin{algorithm}[h]
\begin{algorithmic}[1]
\caption{Agreement($u, v$) \label{alg:agreement-sampling}}
    \IF{$d(u)$ and $d(v)$ are not within factor $1-\beta$}
        \STATE Return ``No''
    \ENDIF
    \STATE {Let $\tau \eqdef \frac{a \cdot \log{n}}{j} \cdot \max\{d(u), d(v)\}$}
    \IF{$X_{u, v} \le 0.9 \cdot \tau$ \alglinelabel{line:testing-value-of-Xuv}}
        \STATE Return ``Yes''
    \ENDIF
    \STATE{Return ``No''}
\end{algorithmic}
\end{algorithm}

We now show that with high probability for every two vertices $u$ and $v$: if the algorithm returns ``Yes'', then $u$ and $v$ are in an agreement; and, if $u$ and $v$ are in $0.8$-weak agreement, then the algorithm returns ``Yes''.

\begin{restatable}{lma}{mpcimplementation}
\label{lemma:MPC-algorithm}
	For any constant $\delta > 0$, there exists an MPC algorithm that, given a signed graph $G = (V, E^+)$, in $O(1)$ rounds for all pairs of vertices $\{u, v\} \in E^+$ outputs ``Yes'' if $u$ and $v$ are in $0.8$-weak agreement, and outputs ``No'' if $u$ and $v$ are not in agreement. Letting $n = |V|$, this algorithm succeeds with probability $1 - 1/n$, uses $n^{\delta}$ memory per machine, and uses a total memory of $\tilde{O}(|E^+|)$.
\end{restatable}

To prove \cref{lemma:MPC-algorithm}, we will use the following well-known concentration inequalities.
\begin{theorem}[Chernoff bound]\label{lemma:chernoff}
	Let $X_1, \ldots, X_k$ be independent random variables taking values in $[0, 1]$. Let $X \eqdef \sum_{i = 1}^k X_i$. Then, the following inequalities hold:
	\begin{enumerate}[(a)]
		\item\label{item:less-than-1} For any $\delta \in [0, 1]$ if $\E{X} \le U$ we have
			\[
				\prob{X \ge (1 + \delta) U} \le \exp\rb{- \delta^2 U / 3}.
			\]
		\item\label{item:lower-tail} For any $\delta > 0$ if $\E{X} \ge U$ we have
			\[
				\prob{X \leq (1 - \delta) U} \le \exp\rb{- \delta^2 U / 2}.
			\]			
	\end{enumerate}
\end{theorem}

\begin{lemma}\label{lemma:u-and-v-when-output=NO}
    Let $u$ and $v$ be two vertices. If \cref{alg:agreement-sampling} returns ``Yes'', then for $a \ge 600$ with probability at least $(1 - n^{-3})$ it holds that $u$ and $v$ are in agreement. (Conversely, the algorithm outputs ``No'' with probability at least $(1 - n^{-3})$ if $u$ and $v$ are not in agreement.)
\end{lemma}
\begin{proof}
    We now upper-bound the probability that $u$ and $v$ are \emph{not} in agreement, but \cref{alg:agreement-sampling} returns ``Yes''.
    
    Assume that $u$ and $v$ are \emph{not} in agreement. Then
    \[
        \E{X_{u, v}} > \tau,
    \]
    where $\tau$ is defined in \cref{alg:agreement-sampling}. (As a reminder, $X_{u, v}$ is defined in \cref{eq:Xuv}.)
    \cref{alg:agreement-sampling} passes the test on Line~\ref{line:testing-value-of-Xuv} with probability
    \[
        \prob{X_{u, v} \le 0.9 \tau}
        \stackrel{\rm{\cref{lemma:chernoff} \eqref{item:lower-tail}}}{\le} \exp\rb{-1/100 \cdot \frac{a \cdot \log{n}}{2}},
    \]
    where we used that $d(u) / j \ge 1$.
    For $a \ge 600$, the last expression is upper-bounded by $n^{-3}$.
\end{proof}

\begin{lemma}
\label{lemma:u-and-v-when-output=YES}
    Let $u$ and $v$ be two vertices that are in $0.8$-weak agreement. Then, for $a \ge 600$ with probability at least $(1 - n^{-3})$ \cref{alg:agreement-sampling} outputs ``Yes''.
\end{lemma}
\begin{proof}
    We have
    \[
        \E{X_{u, v}} \le 0.8 \cdot \tau,
    \]
    where $\tau$ is defined in \cref{alg:agreement-sampling}.
    Hence, \cref{alg:agreement-sampling} outputs ``No'' with probability
    \[
        \prob{X_{u, v} > 0.9 \cdot \tau} \stackrel{\rm{\cref{lemma:chernoff}} \eqref{item:less-than-1}}{\le} \exp{\rb{-1/64 \cdot \frac{a \cdot \log{n}}{3}}},
    \]
    where we used that $d(u) / j \ge 1$. For $a \ge 600$, the last expression is upper-bounded by $n^{-3}$.
\end{proof}
\begin{proof}[Proof of \cref{lemma:MPC-algorithm}]
The implementation part follows by our discussion in \cref{section:implementation} and by having $a = O(1)$.
The claim on probability success follows by using \cref{lemma:u-and-v-when-output=NO,lemma:u-and-v-when-output=YES} and applying a union bound over all $|E^+| \le n^2$ pairs of vertices.
\end{proof}

\subsubsection{Computing connected components}
\label{section:cc}
We now turn to explaining how to compute connected components in $\tG$. Recall that, by \cref{lemma:diameter}, each connected component of $\tG$ has diameter at most $4$. We leverage this fact to design a simple algorithm that in $O(1)$ rounds marks each connected component with a unique id, as follows.

\begin{algorithm}[h]
\begin{algorithmic}[1]
\caption{Connected-Components \label{alg:connected-components}}{}
    \STATE Each vertex $v$ holds an $id_v^i$, $i = 0 \ldots 4$. Let $id_v^0 = v$. 
    \FOR{$i = 1 \ldots 4$}
        \STATE For each $v$, we let $id_v^i = \max_{w \in N(v)} id_w^{i - 1}$
    \ENDFOR
    \STATE Return as a connected component all vertices $w$ that have the same $id_w^4$.
\end{algorithmic}
\end{algorithm}
Let $CC$ be a connected component of $\tG$, and let $\vstar$ be the vertex of $CC$ with the largest label (largest $id^0$).
Correctness of \cref{alg:connected-components} follows by simply noting that at the end of iteration $i$ all the vertices $x$ at distance at most $i$ from $\vstar$ will have $id_x^i = \vstar$. Since $CC$ has diameter at most $4$, it means all the vertices of $CC$ will have the same $id^4$.

\subsubsection{Approximation Analysis}
\label{section:MPC-approximation}
Note that the approximation ratio is affected only by the fact that our algorithm now \emph{estimates} agreement using \cref{alg:agreement-sampling} as opposed to computing it exactly. That is, our MPC algorithm might return that two vertices are not in agreement while in fact they are. Nonetheless, it happens only for vertices which are \emph{not} in $0.8$-weak agreement, i.e., for vertices that are close to not being in agreement.
This might only cause our algorithm to delete more edges;
and the only part of our analysis that suffers from this
are the approximation guarantees of \cref{sec:approximation}.
This can be easily fixed by replacing $\beta$ by $0.8 \cdot \beta$. Then, \cref{theorem:approximation} implies that using \cref{alg:agreement-sampling} to test agreement between vertices still obtains an $O(1)$-approximation.

Now we are ready to prove our main theorem. We restate it for convenience.

\mpcmain*
\begin{proof}
The bounds on the round complexity and memory usage follow directly from the reasoning in \cref{section:computing-agreement,section:cc} and by noticing that step~\ref{line:heavy-and-light-vertices} (determining which vertices are light) can be easily implemented in $O(1)$ MPC rounds.

The approximation guarantees follow because even if we delete some additional edges from $\tG$ that are in agreement but not
in $0.8$-weak agreement, we still obtain a constant-factor approximation as noted above.
\end{proof}

\subsection{Semi-streaming Implementation}
\label{section:streaming-implementation}
We now discuss how to implement our algorithm in the multi-pass semi-streaming setting, and effectively prove \cref{theorem:semi-streaming-main}. In the classic streaming setting, edges of an input graph arrive one by one as a stream. For an $n$-vertex graph, an algorithm in this setting is allowed to use $O(\poly \log{n})$ memory. The semi-streaming setting is a relaxation of the streaming setting, in which an algorithm is allowed to use $O(n \poly \log{n})$ memory. 
We now describe how to implement each of our algorithms in the semi-streaming setting while making multiple passes over the stream. We remark that the order of edges presented in different passes can differ.

To implement \cref{alg:agreement-sampling}, we first fix $O(\log{n})$ random bits for each vertex $v$ and each relevant $j$ (recall that there are $O((\log{n}) / \beta)$ such $j$ values) needed to decide whether $v$ belongs to $S(w, j_w)$, for some $w \in N(v)$.\footnote{As a reminder, $j_v$ is the largest power of $1 / (1 - \beta)$ not greater than $d(v)$.} This is the same as we did in \cref{section:computing-agreement}. Next, we make a single pass over the stream and collect $S(v, j_v)$ and $S(v, j_v / (1 - \beta))$ for each $v$. After this, we are equipped with all we need to compute whether two endpoints of a given edge are in agreement or not.

Next, we make another pass and mark light vertices, where the notion of a light vertex is defined in \cref{alg:main}. Note that bookkeeping which vertices are light requires only $O(n)$ space.

After these steps, in our memory we have (1) a mark for whether each vertex is light or not, and (2) a way to test whether two vertices are in agreement or not without the need to use any information from the stream. This implies that now, whenever an edge arrives on the stream, we can immediately decide whether it belongs to $\tG$ or not. Hence, we have all the information needed to proceed to implementing \cref{alg:connected-components}.

To implement \cref{alg:connected-components}, we make $4$ passes over the stream. In the $i$-th pass, for each edge $\{u, v\}$ on the stream that \emph{belongs} to $\tG$ we update $id_v^i = \max\{id_v^{i - 1}, id_u^{i - 1}\}$ and, similarly for $u$, $id_u^i = \max\{id_v^{i - 1}, id_u^{i - 1}\}$. Since $\tG$ has diameter at most $4$, this suffices to output the desired clusters of $\tG$.

This concludes our implementation of the semi-streaming algorithm.

\section{Empirical Evaluation}
\label{section:experiments}

\begin{table}[h]
\begin{center}
\begin{tabular}{|l|r|r|}
\hline
Graph        & \multicolumn{1}{l|}{\# vertices} & \multicolumn{1}{l|}{\# edges} \\ \hline
dblp-2011    & 986,324                          & 6,707,236                     \\ \hline
uk-2005      & 39,459,925                       & 921,345,078                   \\ \hline
it-2004      & 41,291,594                       & 1,135,718,909                 \\ \hline
twitter-2010 & 41,652,230                       & 1,468,365,182                 \\ \hline
webbase-2001 & 118,142,155                      & 1,019,903,190                 \\ \hline
\end{tabular}
\caption{The datasets used in our experiments. \label{tab:datasets}}
\end{center}
\end{table}

\pgfplotsset{compat=1.7, width=0.8\columnwidth, height=0.4\columnwidth,
/pgfplots/bar cycle list/.style={/pgfplots/cycle list={%
{blue,fill=blue,mark=none},%
{red,fill=red,mark=none},%
{brown,fill=brown,mark=none},%
{green,fill=green!70!black,mark=none},%
{orange,fill=orange,mark=none},%
{blue!30!white,fill=blue!30!white,mark=none},%
{yellow,fill=yellow,mark=none},%
{magenta,fill=magenta,mark=none},%
{cyan,fill=cyan,mark=none},%
{black,fill=black,mark=none},%
}
},
}
\begin{figure}[t!]
\begin{center}
\begin{tikzpicture}[scale=0.6]
\begin{axis}[
font=\Large,
legend style={
                    at={(-0.1,-0.2)},
                    anchor=north west,
                    legend columns=3,
                    /tikz/every even column/.append style={column sep=0.5cm}
                        },
enlargelimits={abs=.2},
ybar=0pt,
bar width=0.15,
xtick={0.5,2.5,4.5,6.5,8.5,10.5},
xticklabels={dblp, uk, it, twitter, webbase},
x tick label as interval,
ylabel = {CorrClustering Obj. Value},
ymin=0.5,
xmin=0.5,
xmax=10.5
]

\addplot+[error bars/.cd,
y dir=both,y explicit]
coordinates {
    (1.5, 1) +- (0.0, 0.000626)
    (3.5, 1) +- (0.0, 0.0002222917098)
    (5.5, 1) +- (0.0, 0.0000818)
    (7.5, 1) +- (0.0, 0.000000214)
    (9.5, 1) +- (0.0, 0.00103)
};
\addplot+[error bars/.cd,
y dir=both,y explicit]
coordinates {
    (1.5, 0.996) +- (0.0, 0.000302) 
    (3.5, 0.977) +- (0.0, 0.00104)
    (5.5, 0.990) +- (0.0, 0.0000210)
    (7.5, 0.999)  +- (0.0, 0.000000408)
    (9.5, 0.986) +- (0.0, 0.00128)
};
\addplot+[error bars/.cd,
y dir=both,y explicit]
coordinates {
    (1.5, 0.978)  +- (0.0, 0.000320)
    (3.5, 0.945) +- (0.0, 0.000362)
    (5.5, 0.978) +- (0.0, 0.000104)
    (7.5, 0.999)  +- (0.0, 0.00000218)
    (9.5, 0.965) +- (0.0, 0.000173)
};
\addplot+[error bars/.cd,
y dir=both,y explicit]
coordinates {
    (1.5, 1.006)  +- (0.0, 0.0227)
    (3.5, 0) +- (0.0, 0.0)
    (5.5, 0) +- (0.0, 0.0)
    (7.5, 0)  +- (0.0, 0.0)
    (9.5, 0) +- (0.0, 0.0)
};
\addplot+[error bars/.cd,
y dir=both,y explicit]
coordinates {
    (1.5, 1.0039)  +- (0.0, 0.0163)
    (3.5, 1.1122)  +- (0.0, 0.0270)
    (5.5, 1.243) +- (0.0, 0.0420)
    (7.5, 1.388)  +- (0.0, 0.09185711641)
    (9.5, 1.22) +- (0.0, 0.0152)
};
\addplot+[error bars/.cd,
y dir=both,y explicit]
coordinates {
    (1.5, 1.009) +- (0.0, 0.0209)
    (3.5, 1.098) +- (0.0, 0.0238)
    (5.5, 1.241) +- (0.0, 0.0334)
    (7.5, 1.295)  +- (0.0, 0.01106967932)
    (9.5, 1.223)  +- (0.0, 0.0176)
};
\addplot+[error bars/.cd,
y dir=both,y explicit]
coordinates {
    (1.5, 0.988)  +- (0.0, 0.0215)
    (3.5, 0)  +- (0.0, 0.0)
    (5.5, 0)  +- (0.0, 0.0)
    (7.5, 0)  +- (0.0, 0.0)
    (9.5, 0)  +- (0.0, 0.0)
};
\addplot+[error bars/.cd,
y dir=both,y explicit]
coordinates {
    (1.5, 0.987)  +- (0.0, 0.00198)
    (3.5, 1.054)  +- (0.0, 0.00820)
    (5.5, 1.244)  +- (0.0, 0.0226)
    (7.5, 1.317)  +- (0.0, 0.01790995734)
    (9.5, 1.189)  +- (0.0, 0.0291)
};
\addplot+[error bars/.cd,
y dir=both,y explicit]
coordinates {
    (1.5, 0.968)  +- (0.0, 0.00265)
    (3.5, 1.057) +- (0.0, 0.00373)
    (5.5, 1.243) +- (0.0, 0.0401)
    (7.5, 1.312)  +- (0.0, 0.0490407152)
    (9.5, 1.197) +- (0.0, 0.0328)
};
\legend{\ouralgoparam{0.05}, \ouralgoparam{0.1}, \ouralgoparam{0.2}, \clusterwparam{0.1},  \clusterwparam{0.5},  \clusterwparam{0.9},  \ppivotparam{0.1}, \ppivotparam{0.5}, \ppivotparam{0.9}}
\end{axis}
\end{tikzpicture}
\caption{The correlation clustering objective values for the different algorithms and configurations that we consider. The objective value of all the algorithms is normalized by dividing by the objective value of OurAlgo0.05 for the respective dataset. \label{fig:quality}}
\end{center}
\end{figure}
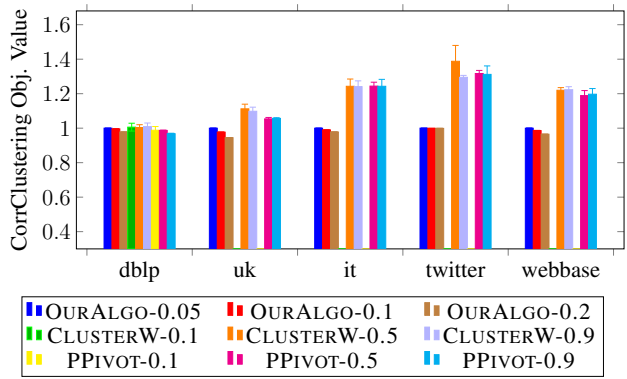

\begin{table*}[t!]\label{fig:experiments-runningtime}
\small
\centering
\begin{tabular}{|r|r|r|r|r|r|r|r|r|r|} 
\hline
 & \multicolumn{3}{c|}{\ouralgo{}} &  \multicolumn{3}{c|}{\clusterw{}} &  \multicolumn{3}{c|}{\ppivot{}} \\
\hline
\backslashbox{Dataset}{param.} & \multicolumn{1}{c|}{$0.05$} & \multicolumn{1}{c|}{$0.1$} & \multicolumn{1}{c|}{$0.2$} & \multicolumn{1}{c|}{$0.1$} & \multicolumn{1}{c|}{$0.5$} & \multicolumn{1}{c|}{$0.9$} & \multicolumn{1}{c|}{$0.1$} & \multicolumn{1}{c|}{$0.5$} & \multicolumn{1}{c|}{$0.9$} \\
\hline
{\datadblp{}} & 1.0\textbf{x}  & 1.1\textbf{x} & 1.0\textbf{x}  & 244.7\textbf{x} & 41.2\textbf{x}  & 18.8\textbf{x}  & 1083.6\textbf{x} & 119.5\textbf{x} & 42.7\textbf{x}\\
\hline
{\datauk{}} & 5.5\textbf{x}  & 6.5\textbf{x}  & 10.7\textbf{x} & -   & 445.5\textbf{x} & 213.1\textbf{x} & -    & 490.8\textbf{x} & 217.4\textbf{x}\\
\hline
{\datait{}} &10.5\textbf{x} & 14.8\textbf{x} & 12.4\textbf{x} & -   & 475.7\textbf{x} & 290.8\textbf{x} & -    & 762.8\textbf{x} & 274.9\textbf{x} \\
\hline
{\datatwitter{}} & 8.8\textbf{x}  & 15.5\textbf{x} & 13.9\textbf{x} & -   & 837.5\textbf{x} & 300.2\textbf{x} & -    & 730.2\textbf{x} & 392.8\textbf{x} \\
\hline
{\datawebbase{}} & 13.0\textbf{x} & 13.5\textbf{x} & 14.6\textbf{x} & -   & 835.1\textbf{x} & 436.8\textbf{x} & -    & 789.3\textbf{x} & 458.1\textbf{x} \\
\hline
\end{tabular}	
\caption{Average running times for the algorithms (with different parameters) that we consider. All times are reported relative to the execution time of \ouralgoparam{0.05} on the dataset \datadblp{}, which is approximately 21 seconds. We use 10 machines. \label{tab:hi-runtime}}
\end{table*}

\paragraph{Datasets.}To empirically analyze our algorithm compared to state-of-the-art parallel algorithms for correlation clustering, we considered a collection of two social networks and three web graphs.
All our datasets were obtained from The Laboratory for Web Algorithmics\footnote{http://law.di.unimi.it/datasets.php} \cite{BoVWFI,BRSLLP,BCSU3}, and some of their statistics are summarized in ~\cref{tab:datasets}. The \emph{dblp-2011} dataset is the DBLP co-authorship network from 2011, \emph{uk-2005} is a 2005 crawl of the .uk domain, \emph{it-2004} a 2004 crawl of the .it domain, \emph{twitter-2010} a 2010 crawl of twitter, and \emph{webbase-2001} is a 2001 crawl by the WebBase crawler. We converted all datasets to be undirected and removed parallel edges. The correlation clustering instance is formed by considering all present edges as ``+'' edges and all missing edges as ``-'' edges.

\paragraph{Algorithms and parameters.}
In our experiments we consider three algorithms: our algorithm from \cref{sec:algorithm} (we refer to it as \ouralgo{}), as well as the ClusterWild (\clusterw{}, in short) algorithm from \citet{pan2015parallel} and the ParallelPivot (\ppivot{}, in short) from \citet{chierichetti2014correlation}. \clusterw{} and \ppivot{} admit a parameter $\epsilon$, which affects the number of parallel rounds required to perform the computation, depending on the structure of the input graph. For \ppivot{}, $\epsilon$ also slightly affects the theoretical approximation guarantees (i.e., the approximation is $(3+\epsilon)$). We adopt the setting of $\epsilon$ from \citet{pan2015parallel}, and use $\epsilon\in\{0.1, 0.5, 0.9\}$ for both algorithms. Our algorithm has two parameters $\lambda,\beta$ which affect the approximation of the algorithm (see \cref{lemma:cut-step-1} and \cref{lemma:cut-step-3}), but the number of rounds is independent of these parameters and is a fixed constant. For simplicity, we set $\lambda=\beta \in \{0.05, 0.1, 0.2\}$. 
To refer to an algorithm with a specific parameter, we append the parameter value to the algorithm name, e.g., we say \ouralgoparam{0.05}.

\paragraph{Implementation details.}
In all our experiments the vertices are randomly partitioned among machines (we note that no algorithm requires a fixed partitioning of the input vertices onto machines). 
We made a fair effort to implement all algorithms equally well, and we did not use any tricks or special data structures.
For simplicity, we assume that the entire neighborhood of each vertex fits on a single machine (this is not required by any algorithm). Removing this assumption would increase the number of rounds of all algorithms by a constant factor and most likely would not significantly affect their relative running times. 

\paragraph{Setup and methodology.}
We used 10 machines across all experiments (except for \cref{sec:speed-up}); this is enough for the machines to collectively fit the input graph in memory.
We repeated all experiments 3 times, and we report relative average running times (wall-clock time),
as a ratio of each measurement compared to the minimum average running time observed across our experiments. 
We did not use a dedicated system for our experiments.
Executions that were running for an unreasonable amount of time (more than 72 hours) were stopped, and we report no data for such executions; these occurred only for \clusterwparam{0.1} and \ppivotparam{0.1}. 
We excluded the time that it takes to load the input graph into the memory, as this is unavoidable and uniform across all algorithms.

\subsection{Results on quality}
\cref{fig:quality} summarizes the results of our experiments in terms of solution quality, that is, the correlation clustering objective value of the solution computed by the algorithms that we consider. \ouralgo{} consistently produces better solutions compared to the two competitor algorithms \clusterw{} and \ppivot{}. 
In particular, for all datasets but \datadblp{}, \clusterw{} and \ppivot{} produce solutions whose numbers of disagreements are more than 10\% to 30\% higher compared to the best solution produced by \ouralgo{}. For \datadblp{}, our \ouralgo{} is very comparable but slightly better than the baselines.

In terms of variance in the quality of the produced clustering between the different runs, \ouralgo{} has negligible variance, which is natural given that the only source of randomness comes from identifying pairs of vertices that are in agreement. On the other hand, the behavior of \clusterw{} and \ppivot{} is not as stable, in terms of the quality of the produced solution, as demonstrated by the standard deviation illustrated in \cref{fig:quality}.

Moreover, \cref{fig:quality} shows that the behavior of \ouralgo{} is not very sensitive to the choice of the parameters $\lambda, \beta$, as for all settings of these parameters \ouralgo{} produces solutions that are significantly better compared to the state-of-the-art parallel algorithms for correlation clustering. Recall that the parameter $\epsilon$ in \clusterw{} does not affect the solution quality, while it only slightly affects the theoretical guarantees of \ppivot{}. In our experiments we did not observe any correlation between the choice of $\epsilon$ and the quality of the solution produced by \ppivot{}.

\subsection{Performance results}
We summarize the average running times of the different algorithms in~\cref{tab:hi-runtime}. For each algorithm, we report
the ratio of its average running time to the average running time of \ouralgoparam{0.05} on the \datadblp{} dataset, which is the fastest average running time we observed throughout our experiments, equal to roughly 21 seconds. It is evident that \ouralgo{} (independently of its parameters) is consistently over an order of magnitude faster compared to the state-of-the-art parallel algorithms \clusterw{} and \ppivot{}, and in several cases the gap increases to two orders of magnitude.

While the choice of $\lambda,\beta$ in \ouralgo{} has no effect on the number of rounds performed by \ouralgo{}, one can observe some deviations between the different parameter choices, which is likely due to time-specific system work-load. Nonetheless,
for each algorithm its maximum running time across all runs is within a factor at most 2 of its average running time.
While the same can be said for \clusterw{} and \ppivot{}, throughout our experiments we did not observe any case where an execution of either of \clusterw{} or \ppivot{} performed within a factor 10 of any execution of \ouralgo{}, even for the smallest instance \datadblp{}, where the running times are expected to be the closest. 
On the other hand, the choice for the parameter $\epsilon$ affects the running time of \clusterw{} and \ppivot{} and requires
proper tuning depending on the structure of the input graph (in our graphs, the choice of $\epsilon=0.9$ always results in
significantly faster performance compared to other choices).
The executions of \clusterw{} and \ppivot{} with $\eps=0.1$, on all datasets except \datadblp{}, were stopped as they did not terminate within a reasonable amount of time, and thus are not reported.

\subsection{Speedup Evaluation} \label{sec:speed-up}
\begin{table*}
\small
\centering
\begin{tabular}{|r|r|r|r|r|r|} 
\hline
\backslashbox{Dataset}{\#machines} & \multicolumn{1}{c|}{1} & \multicolumn{1}{c|}{2} & \multicolumn{1}{c|}{4} & \multicolumn{1}{c|}{8} & \multicolumn{1}{c|}{16}  \\
\hline
{\datait{}} &1\textbf{x}  ($\pm$0.303)& 2.392\textbf{x}  ($\pm$0.032)& 3.114\textbf{x}  ($\pm$0.119) & 4.823\textbf{x}  ($\pm$0.266) & 5.445\textbf{x} ($\pm$0.790)  \\
\hline
{\datatwitter{}} & 1\textbf{x} ($\pm$0.149) & 4.451\textbf{x} ($\pm$0.0151) & 4.968\textbf{x} ($\pm$0.0803) & 7.270\textbf{x} ($\pm$0.138) & 5.479\textbf{x} ($\pm$0.338)  \\
\hline
{\datawebbase{}} & 1\textbf{x} ($\pm$0.0618) & 5.280\textbf{x} ($\pm$0.225) & 4.441\textbf{x} ($\pm$0.166) & 12.161\textbf{x} ($\pm$0.0110) & 11.306\textbf{x} ($\pm$0.047)  \\
\hline
\end{tabular}	
\caption{Average speedup achieved by \ouralgoparam{0.05}, for an increasing number of machines. The standard deviation of the running time, as a fraction of the running time, is presented in parentheses. \label{tab:speedup}}
\end{table*}

In this section we study the parallelism of \ouralgo{}. We use a fixed parameter $\lambda=\beta=0.05$, as the choice of this parameter does not significantly affect the running time of the algorithm; indeed, when repeating the experiments for different parameter settings, we observed a very similar picture to the one we report below. 
To measure speed-up, we start from $1$ machine and we double the number of machines at each step, that is, we consider $1$, $2$, $4$, $8$, and $16$ machines.
Each reported running time is the average time of three repetitions of the algorithm, presented relative to the average running time of \ouralgoparam{0.05} with $1$ machine. Our results are summarized in \cref{tab:speedup}. 

Across all datasets, we observe a trend of near-linear speedup as the number of machines grows from $1$ to $8$. There is no significant speedup in the transition from $8$ to $16$ (in fact, in two out of the three cases we see worse running times when using $16$ machines), and this is likely because we reach a tipping point where the cost of communication between the machines is higher compared to the benefit gained by parallelism, for the specific datasets that we consider. Moreover, the speedup achieved across the three datasets is not uniform, and this is due to the fact that $1$ machine might be more appropriate for some datasets but not enough for other datasets; indeed, the highest speedup is achieved for the \datawebbase{} dataset, which the largest among the graphs that we consider. 

Although we observe small inconsistencies in the overall picture of our experiment, which is due to high variance in the observed running times (recall that we do not use a dedicated system for our experiments), one can observe a clear trend highlighting a near-linear speedup  as the number of machines increases.

\subsection{Cluster Statistics and Number of Rounds}
In this section, we provide various statistics regarding the performance and solutions produced by the algorithms. \cref{fig:cluster-size-distributions} presents the distribution of the cluster sizes. We observe that the size of the clusters are smaller in  \ouralgo{} compared to the baselines. Evidently, this is due to the fact that \ouralgo{} produces only dense clusters, as opposed to \clusterw{} and \ppivot{} which often produce very sparse clusters. \cref{tab:clusterstat} indicates that the datasets for which the clusters produced by \clusterw{} and \ppivot{} are the sparsest are the datasets for which the distributions of the cluster sizes differ the most between \ouralgo{} and \clusterw{} (or \ppivot{}).

\cref{tab:clusterstat} presents the number of MPC rounds required by each algorithm,  the number of clusters in each solution and the number of existing intra-cluster edges for each solution. We observe that \ouralgo{} requires a fixed number of MPC rounds that is significantly smaller (up to a factor $90$) compared to \clusterw{} and \ppivot{}. Moreover, while \ouralgo{} produces solutions with more clusters compared to \clusterw{} and \ppivot{}, the produced clusters are much denser than those produced by \clusterw{} and \ppivot{}.

\begin{figure}[h!]
     \centering
     \begin{subfigure}[b]{0.49\textwidth}
         \centering
         \includegraphics[width=\textwidth]{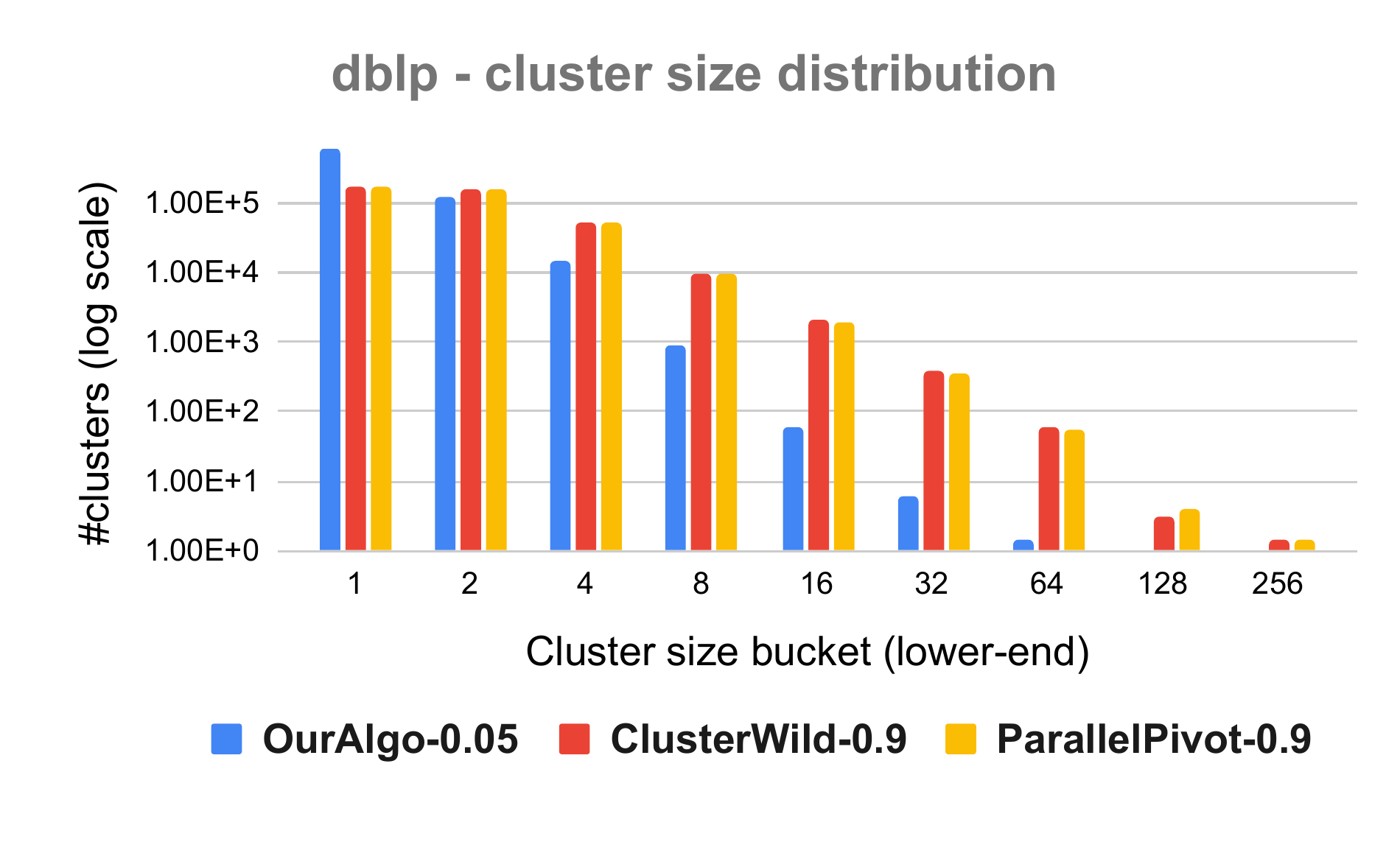}
         \label{fig:y equals x}
     \end{subfigure}
     \hfill
     \begin{subfigure}[b]{0.49\textwidth}
         \centering
         \includegraphics[width=\textwidth]{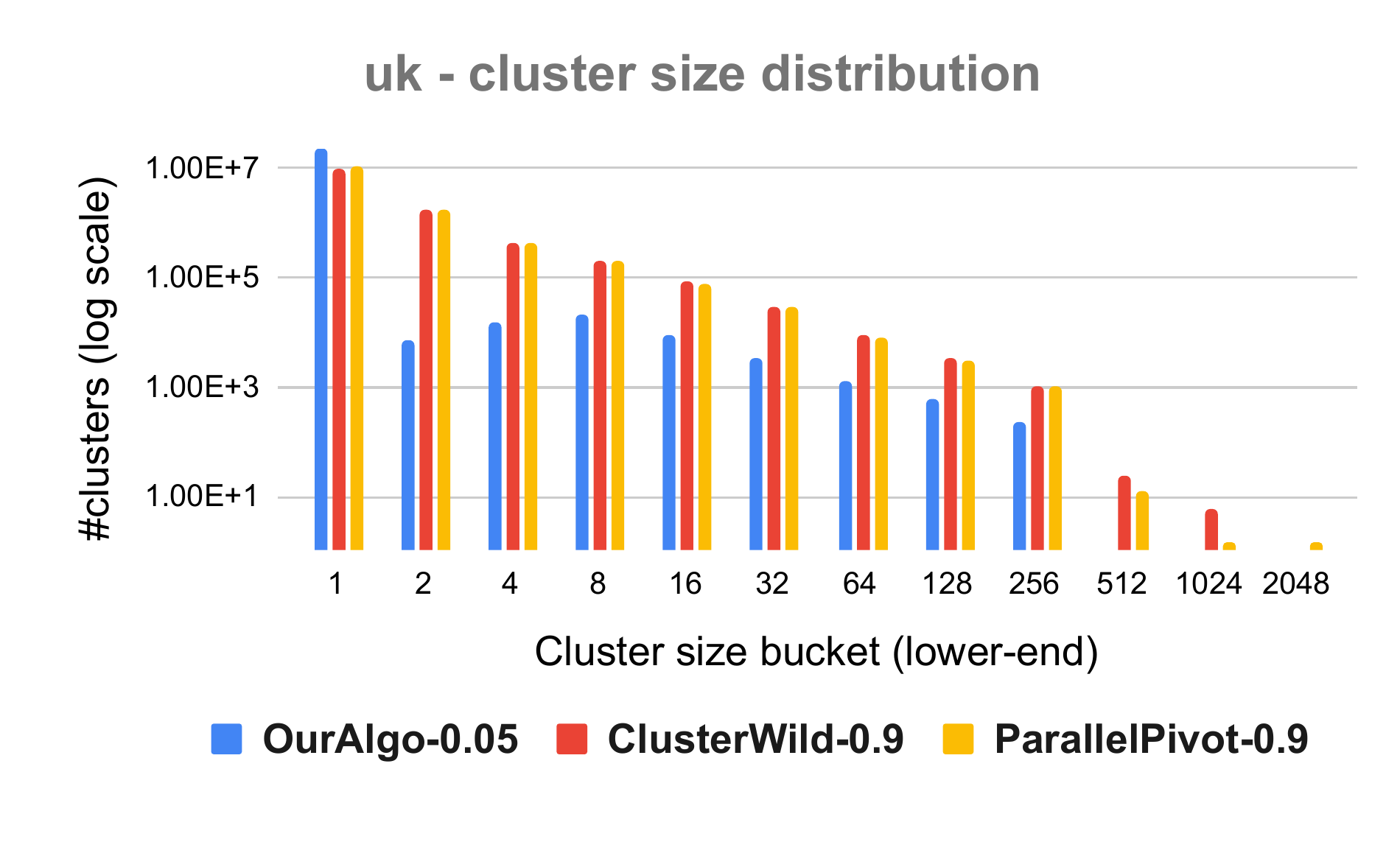}
         \label{fig:three sin x}
     \end{subfigure}
     \hfill
     \begin{subfigure}[b]{0.49\textwidth}
         \centering
         \includegraphics[width=\textwidth]{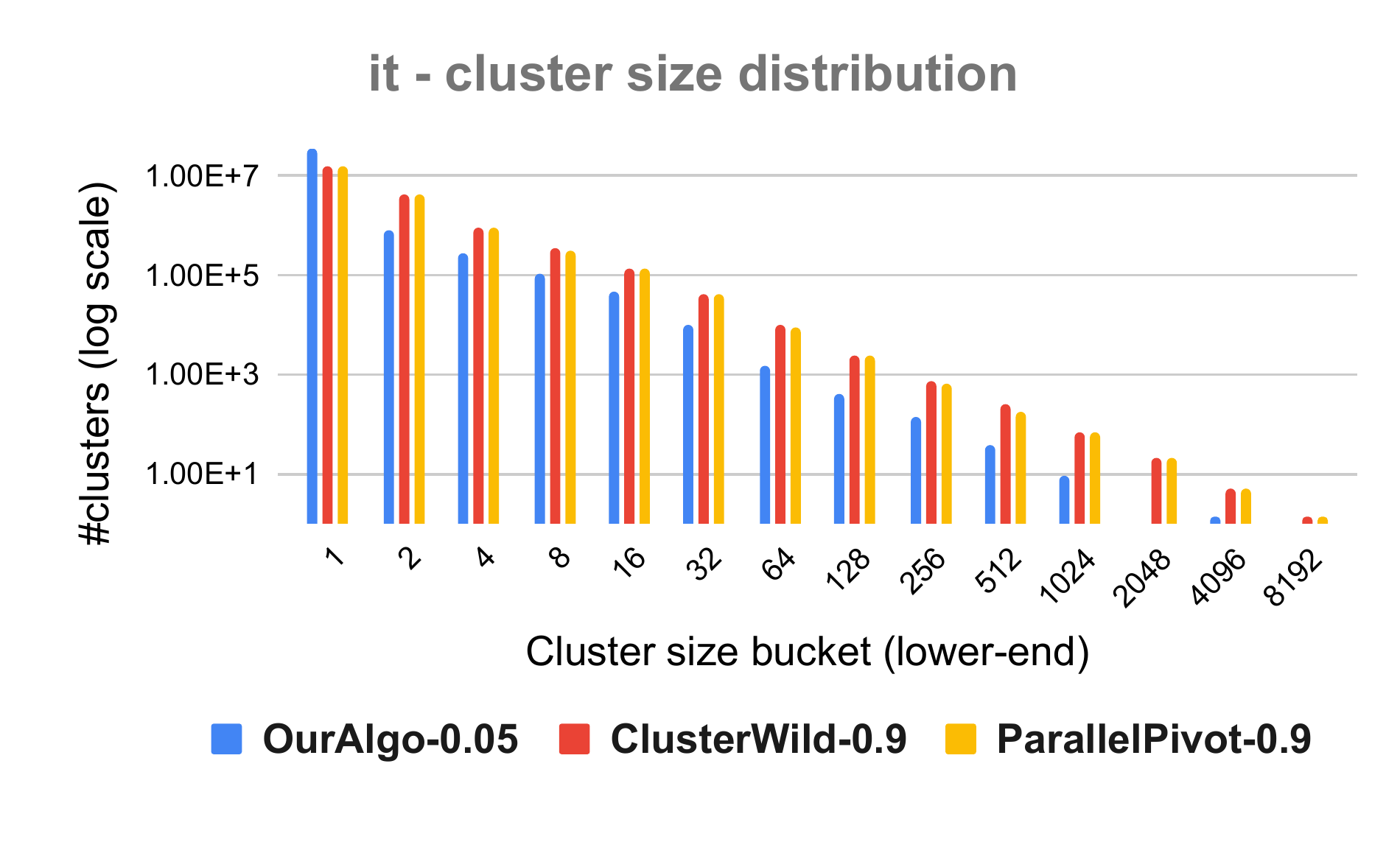}
         \label{fig:five over x}
     \end{subfigure}
     \hfill
     \begin{subfigure}[b]{0.49\textwidth}
         \centering
         \includegraphics[width=\textwidth]{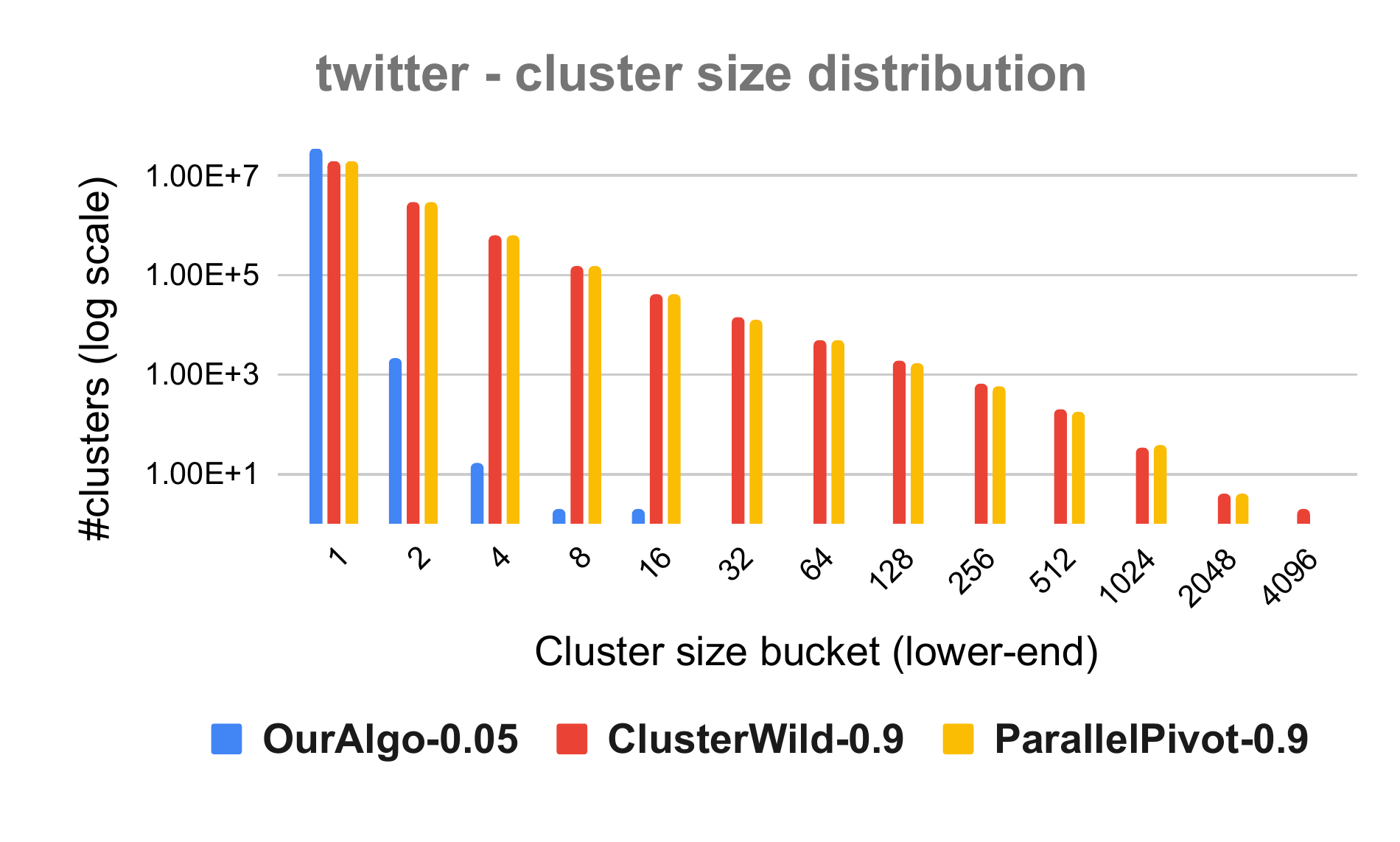}
         \label{fig:five over x}
     \end{subfigure}
     \hfill
     \begin{subfigure}[b]{0.49\textwidth}
         \centering
         \includegraphics[width=\textwidth]{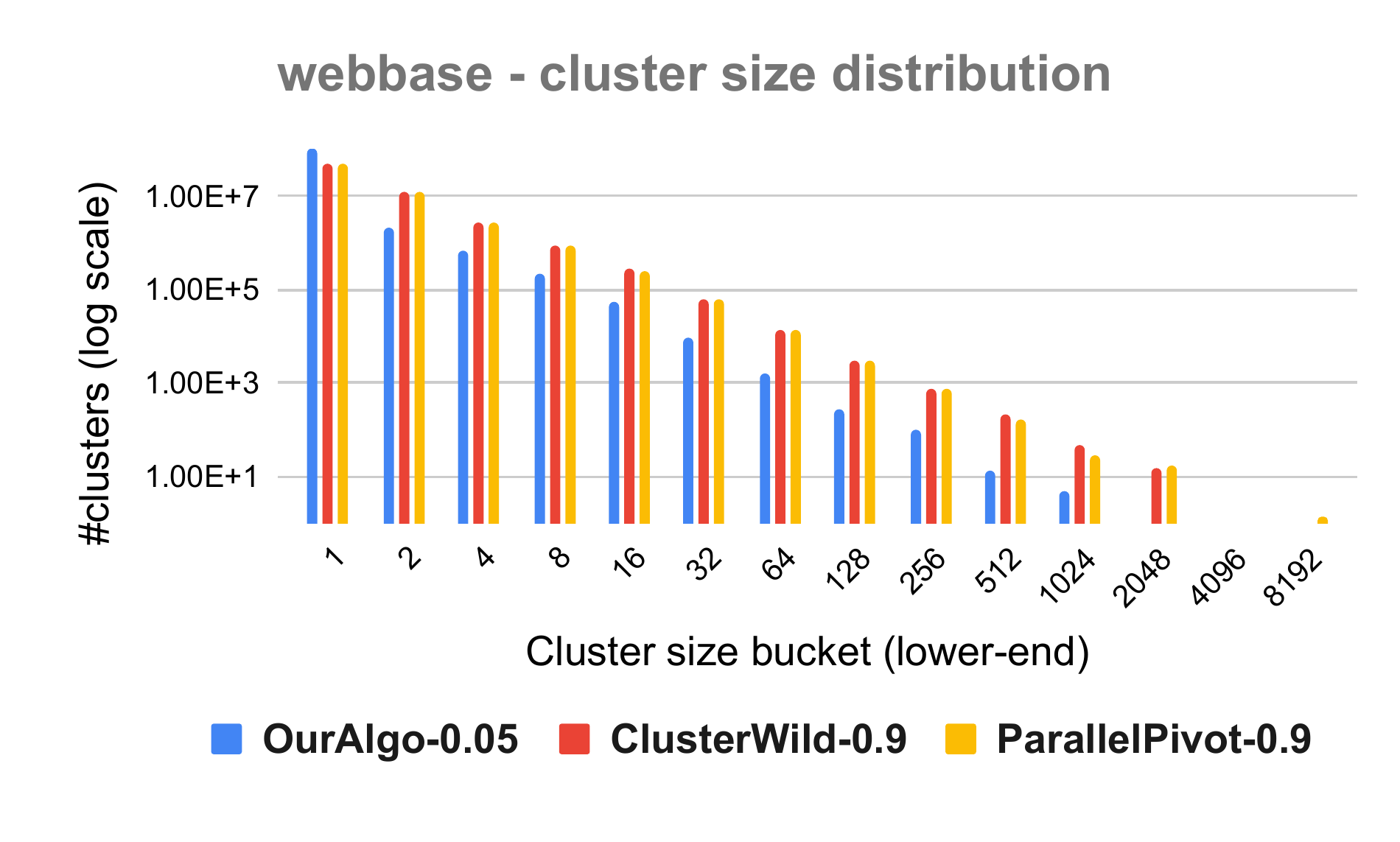}
         \label{fig:five over x}
     \end{subfigure}
     \vspace*{-0.5em}
        \caption{The cluster size distributions produced by the algorithms \ouralgoparam{0.05}, \clusterwparam{0.9}, and \ppivotparam{0.9} on all datasets that we considered.}
        \label{fig:cluster-size-distributions}
\end{figure}

\begin{table}[h!]
\small
\centering
\begin{tabular}{rrrrrrrlll}
\cline{2-10}
\multicolumn{1}{r|}{} & \multicolumn{3}{c|}{\datadblp{}} & \multicolumn{3}{c|}{\datauk{}} & \multicolumn{3}{c|}{\datait{}} \\ \cline{2-10} 
\multicolumn{1}{l|}{} & \multicolumn{1}{c|}{\#rounds} & \multicolumn{1}{c|}{\#clusters} & \multicolumn{1}{c|}{in-edges} & \multicolumn{1}{c|}{\#rounds} & \multicolumn{1}{c|}{\#clusters} & \multicolumn{1}{c|}{in-edges} & \multicolumn{1}{c|}{\#rounds} & \multicolumn{1}{c|}{\#clusters} & \multicolumn{1}{c|}{in-edges} \\ \hline
\multicolumn{1}{|r|}{\ouralgoparam{0.05}} & 33 & 723,511 & \multicolumn{1}{r|}{1.000} & 33 & 22,999,216 & \multicolumn{1}{r|}{0.955} & \multicolumn{1}{r}{33} & \multicolumn{1}{r}{36,467,636} & \multicolumn{1}{r|}{0.972} \\ \hline
\multicolumn{1}{|r|}{\ouralgoparam{0.1}} & 33 & 720,229 & \multicolumn{1}{r|}{0.999} & 33 & 22,764,081 & \multicolumn{1}{r|}{0.933} & \multicolumn{1}{r}{33} & \multicolumn{1}{r}{34,244,835} & \multicolumn{1}{r|}{0.957} \\ \hline
\multicolumn{1}{|r|}{\ouralgoparam{0.2}} & 33 & 704,489 & \multicolumn{1}{r|}{0.996} & 33 & 22,228,865 & \multicolumn{1}{r|}{0.895} & \multicolumn{1}{r}{33} & \multicolumn{1}{r}{31,042,932} & \multicolumn{1}{r|}{0.735} \\ \hline
\multicolumn{1}{|r|}{\clusterwparam{0.9}} & 725 & 382,491 & \multicolumn{1}{r|}{0.516} & 1441 & 12,778,648 & \multicolumn{1}{r|}{0.461} & \multicolumn{1}{r}{1837} & \multicolumn{1}{r}{22,457,586} & \multicolumn{1}{r|}{0.287} \\ \hline
\multicolumn{1}{|r|}{\ppivotparam{0.9}} & 1160 & 386,275 & \multicolumn{1}{r|}{0.537} & 2280 & 12,944,056 & \multicolumn{1}{r|}{0.452} & \multicolumn{1}{r}{2610} & \multicolumn{1}{r}{22,675,174} & \multicolumn{1}{r|}{0.316} \\ \hline
\multicolumn{1}{l}{} & \multicolumn{1}{l}{} & \multicolumn{1}{l}{} & \multicolumn{1}{l}{} & \multicolumn{1}{l}{} & \multicolumn{1}{l}{} & \multicolumn{1}{l}{} &  &  &  \\ \cline{2-7}
\multicolumn{1}{l|}{} & \multicolumn{3}{c|}{\datatwitter{}} & \multicolumn{3}{c|}{\datawebbase{}} &  \multicolumn{1}{c}{} & \multicolumn{1}{c}{} & \multicolumn{1}{c}{} \\ \cline{2-7}
\multicolumn{1}{l|}{} & \multicolumn{1}{c|}{\#rounds} & \multicolumn{1}{c|}{\#clusters} & \multicolumn{1}{c|}{in-edges} & \multicolumn{1}{c|}{\#rounds} & \multicolumn{1}{c|}{\#clusters} & \multicolumn{1}{c|}{in-edges} &  &  &  \\ \cline{1-7}
\multicolumn{1}{|r|}{\ouralgoparam{0.05}} & 33 & 34,981,120 & \multicolumn{1}{r|}{0.990} & 33 & 106,613,511 & \multicolumn{1}{r|}{0.988} &  &  &  \\ \cline{1-7}
\multicolumn{1}{|r|}{\ouralgoparam{0.1}} & 33 & 34,980,638 & \multicolumn{1}{r|}{0.990} & 33 & 103,908,793 & \multicolumn{1}{r|}{0.957} &  &  &  \\ \cline{1-7}
\multicolumn{1}{|r|}{\ouralgoparam{0.2}} & 33 & 34,978,139 & \multicolumn{1}{r|}{0.973} & 33 & 99,049,622 & \multicolumn{1}{r|}{0.866} &  &  &  \\ \cline{1-7}
\multicolumn{1}{|r|}{\clusterwparam{0.9}} & 1876 & 24,572,801 & \multicolumn{1}{r|}{0.077} & 1721 & 68,800,036 & \multicolumn{1}{r|}{0.346} &  &  &  \\ \cline{1-7}
\multicolumn{1}{|r|}{\ppivotparam{0.9}} & 2580 & 24,701,912 & \multicolumn{1}{r|}{0.068} & 2510 & 69,394,341 & \multicolumn{1}{r|}{0.331} &  &  &  \\ \cline{1-7}
\end{tabular}	
\caption{\label{tab:clusterstat}This table presents the number of MPC rounds (\#rounds), number of clusters (\#clusters) and the fraction of intra-cluster edges found in each solution (in-edges).}
\end{table}

\section{Conclusions and Future Work}
We present a new parallel algorithm for correlation clustering and we prove both theoretically and experimentally that our algorithm is extremely fast and returns high-quality solutions. Interesting open problems are to improve the approximation guarantees of our algorithm and to establish a more formal connection between our results and well-known similar heuristics~\cite{xu2007scan}. Another direction would be to design an MPC algorithm in the sublinear regime for the \emph{weighted}
version of the problem.

\section*{Acknowledgments}
We thank the anonymous reviewers for their valuable comments.
S.~Mitrovi\' c was supported by the Swiss NSF grant No.~P400P2\_191122/1,  MIT-IBM Watson AI Lab and research collaboration agreement No. W1771646, NSF award CCF-1733808, and FinTech@CSAIL.
\bibliographystyle{icml2021}
\bibliography{bibliography}

\end{document}